\documentclass[12pt,reqno,pre,longbibliography]{amsart}
\usepackage{amsfonts}
\usepackage{graphicx}
\usepackage{amscd}
\usepackage{amsmath,amssymb}
\usepackage{epsfig}
\usepackage[usenames]{color}
\usepackage{subfigure}
\usepackage[english]{babel}

\providecommand{\U}[1]{\protect\rule{.1in}{.1in}}
\providecommand{\U}[1]{\protect\rule{.1in}{.1in}}
\allowdisplaybreaks
\newtheorem{theorem}{Theorem}

\newtheorem{proposition}{Proposition}
\theoremstyle{plain}
\newtheorem{acknowledgement}{Acknowledgement}

\numberwithin{equation}{section}

\DeclareMathOperator{\sech}{sech}

\begin{document}
\title[Generalized Coupled Systems]{On explicit solutions for coupled reaction-diffusion and  Burgers-type equations with variable coefficients through a Riccati System}
\author{José M. Escorcia }
\address{Escuela de Ciencias Aplicadas e Ingenier\'ia, Universidad EAFIT, Carrera 49 No. 7 Sur-50, Medellín 050022, Colombia.}
\email{jmescorcit@eafit.edu.co}
\author{E. Suazo}
\address{School of Mathematical and Statistical Sciences, University of
Texas Rio Grande Valley, 1201 W. University Drive, Edinburg, Texas,
78539-2999.}
\email{erwin.suazo@utrgv.edu}

\date{\today }


\begin{abstract}
This work is concerned with the study of  explicit solutions for generalized coupled reaction-diffusion and Burgers-type systems  with variable coefficients. Including nonlinear models with variable coefficients such as diffusive Lotka-Volterra model, the Gray-Scott model, the Burgers equations. The equations' integrability (via the explicit formulation of the solutions) is accomplished by using similarity transformations and requiring that the coefficients fulfill a Riccati system. We present traveling wave type solutions, as well as solutions with more complex dynamics and relevant features such as bending. A Mathematica file has been prepared as supplementary material, verifying the Riccati systems used in the construction of the solutions.  \\


\end{abstract}

\maketitle

\section{Introduction}

A large number of problems in engineering and applied sciences are modeled by partial differential equations (PDEs) or systems of PDEs. In this regard, the study of such PDEs or systems of PDEs is a constantly evolving area due to their importance in several fields of knowledge. One type of differential equations that has attracted the attention of researchers in recent years are those of reaction-diffusion. In biology, the reaction-diffusion systems explain the pattern of stripes of leopards, snakes, zebras, and jaguars \cite{Liu2006,HIROTO2002,Turing1990} and creates a model for replicating the hepatitis B virus \cite{Ghaemi2023,Tamko2023,Wang2007}. Reaction-diffusion equations can also be used to study tumor growth \cite{Chaplain1995,Salman2003,Yanying2022} and fission wave dynamics \cite{Andrew2021}, among other topics.

From a theoretical standpoint, aspects of reaction-diffusion models, such as well-posedness and asymptotic behavior of the solutions, have also been thoroughly explored \cite{Mark1991,Mebarki2023,Morgan2023,Sleeman1987,SOMATHILAKE2005}, allowing for significant progress in our understanding of these systems. A hot topic in the study of PDEs and with great impact on both theoretical and applied domains is the development of techniques that allow exact solutions \cite{Polyanin2019, Polyanin2019_2, Polyanin2021} . In the case of reaction-diffusion equations, there is a wealth of literature describing various strategies; a few instances include the Hirota method \cite{Oktay2005,Gamze2006}, similarity transformations \cite{Enrique2018,Suazo13}, and the tanh-method  \cite{Malfliet1996,Rodrigo2000} among others \cite{Aibinu2023,Ali2023,Polyanin2019,Polyanin2022}, that have enabled the construction of explicit solutions. 

The aforementioned methods have been successfully applied to classical models such as the linear coupled  reaction-diffusion equation \cite{Ellery2013,Simpson2015}, the diffusive Lotka-Volterra system  \cite{Roman2022,Cherniha2004,Lu2022}, the Gray-Scott model \cite{GRAY1984,Navneet2022,RODRIGO2001}, and the coupled Burgers system \cite{Abazari2010,Mittal2011,Abdul2007} (Section \ref{Sect2} provides additional information on these systems). Essentially, the type of solutions most reported in the literature for these models are the so-called traveling waves, due to their fascinating mathematical and physical properties. However, there is still a lot of interest in finding solutions with more complicated structures (compared to solitary waves) that allow us to describe phenomena that the former cannot explain. The main goal of this research  is to introduce some generalizations of these models as well as a mechanism for finding explicit solutions.

To be more specific, we primarily study the explicit solutions for the general coupled equations of the type
\begin{eqnarray}
    \psi_t &=& H(\psi) + F_1(\psi,\varphi), \label{Gen1} \\
  \varphi_t &=& H(\varphi) + F_2(\psi,\varphi), \label{Gen2}
\end{eqnarray}
where the variable operator $H$ has a quadratic form 
$$ H(\psi) := a(t)\psi_{xx} -b(t)x^2 \psi + c(t)x\psi_x + d(t)\psi + xf(t)\psi -g(t)\psi_x$$
and $F_i(\psi,\varphi)$ ($i = 1,2$) will be given. Here, the coefficients $a(t),$ $b(t),$ $c(t),$ $ d(t),$ $ f(t),$ $g(t)$ are time-dependent real functions. By selecting appropriate $F_i$, the system (\ref{Gen1})-(\ref{Gen2}) can be viewed as generalizations of the classical reaction-diffusion models and the Burgers system discussed above. The integrability of equations (\ref{Gen1})-(\ref{Gen2}) is accomplished by requiring the coefficients to fulfill a Riccati system and extending, some classical model solutions via similarity transformations. Although the solutions discovered are primarily derived from traveling waves, we emphasize that our findings are considerably more general and applicable to any sort of solution. However, the employment of traveling waves did not constitute a limitation since it was clear from these examples that the system (\ref{Gen1})-(\ref{Gen2}) allowed solutions of the same kind, as well as others with more fascinating dynamics. 

The Riccati system (\ref{Rica_1})-(\ref{Carac__1}) has been applied to the analysis of several PDEs with variable coefficients, including the nonlinear Schr{\"o}dinger equation \cite{Acosta-Humanez,Suazo16,CorderoSoto2008,Escorcia,Suazo,Suazo18}, diffusion-type equations \cite{Suazo13}, and reaction-diffusion equations \cite{Enrique2018}. In our recent paper \cite{Jose2023}, we use the Riccati system to investigate the existence (by explicit construction) of Rogue wave and dark-bright soliton-like solutions for a coupled system of Schr{\"o}dinger equations (generalizing the Manakov system). And now in the present work, we extend the ideas of \cite{Enrique2018} by investigating a coupled reaction-diffusion system and coupled Burgers equations with the structure indicated in (\ref{Gen1})-(\ref{Gen2}).

After the system (\ref{Rica_1})-(\ref{Carac__1}) was introduced in \cite{CorderoSoto2008} it has brought attention of different researchers. In  \cite{CARINENA20142303} it was considered a new example of a Lie system that is not a Hamiltonian system. It is also one of the few Lie systems connected to PDEs until now \cite{CARINENA20142303}. System (\ref{Rica_1})-(\ref{Carac__1}) describes integral curves of the t-dependent vector field of the form
$$
X_t=a(t)X_1+b(t)X_2+c(t)X_3-2d(t)X_4+f(t)X_5+g(t)X_6
$$
where $X_1,...,X_6$  satisfy certain commutations rules showing that (\ref{Rica_1})-(\ref{Carac__1}) is a Lie system associated to a Vessiot Guldberg Lie algebra of Hamiltonian vector fields with respect to a presymplectic form, see \cite{CARINENA20142303}. Systems such as (\ref{Rica_1})-(\ref{Carac__1}) originated the introduction of a particular class of Lie systems on Dirac manifolds, denominated Dirac–Lie systems, see \cite{CARINENA20142303}. It is part of the objectives of this work show more applications of this Lie system in the construction of transformations for nonlinear reaction–diffusion equations.

 The consideration of parameters in this work is inspired by the work of Marhic, who, in 1978 \cite{Marhic1978} introduced (probably for the first time) a one-parameter $\{\alpha(0)\}$ family of solutions for the linear Schrödinger equation of the one-dimensional harmonic oscillator. The solutions presented by Marhic constituted a generalization of the original Schrödinger wave packet with oscillating width. 

We have also introduced solutions for a modified Riccati system (\ref{Ricati1})-(\ref{Ricati7}), see the appendix. We have prepared a Mathematica file as supplementary material verifying the Riccati systems used in this work.

This work makes the following contributions: 
\begin{itemize}
    \item{The introduction of \emph{new coupled reaction-diffusion and  Burgers-type equations with variable coefficients} and, by the use of similarity transformations and  Riccati system, we provide a mechanism to find explicit solutions. These new models generalize some classical systems such as the linear coupled reaction-diffusion equation, the diffusive Lotka-Volterra system, the Gray-Scott model, and the coupled Burgers system. }
    \item{The report of \emph{new solutions to coupled reaction-diffusion and Burgers-type equations}, include traveling wave solutions as well as solutions with more advanced kinetics determined by bending properties.} 
    \item{The introduction of a \emph{modified Riccati system} that allows us to generate exponential-type solutions for the general linear reaction-diffusion system. The solution for this Riccati system is provided. All solutions of the Riccati systems utilized in the current research were verified using the computer algebra system Mathematica.} 
\end{itemize}

The rest of the paper is organized as follows: In Section \ref{Sect2}, we present a brief overview of classical models, such the linear coupled reaction-diffusion equation, the diffusive Lotka-Volterra system, the Gray-Scott model, and the coupled Burgers system. Some of the solutions found in the literature are also presented in this section. Section \ref{Sect3} is the core  of the current research. Generalizations of the systems discussed in the preceding section are presented here, as well as sufficient conditions for the construction of exact solutions in terms of the solutions of the constant coefficient models (see Theorems \ref{Th1}-\ref{Th4} and Proposition \ref{PROPOSITION1}). In particular, Theorem \ref{Th1_1} employs the modified Riccati system. As a result of the coefficient balance requirements (Riccati system), traveling wave type solutions and solutions with noticeable bending dynamics are shown. A Mathematica file has been prepared as supplementary
material verifying the Riccati systems used in the construction of the solutions in this section. Conclusions and some final remarks are given in Section \ref{Sect4}. Section \ref{Sect5} corresponds to the
appendix, in which the solutions of the Riccati systems are presented.

\section{Classical Coupled Reaction-Diffusion and Burgers Systems} \label{Sect2}
In this section, we briefly describe some well-known reaction-diffusion equations, such as: the linear reaction-diffusion model, the diffusive Lotka-Volterra, and the Gray-Scott model. Likewise, the coupled Burgers equations are presented. We mainly exhibit the traveling wave solutions of these equations. However, the findings of the paper are valid for all  solutions of these classical systems.  
\subsection{The linear reaction-diffusion system} Several processes during the embryonic state are associated with the migration and proliferation of cells within growing tissues. A canonical model for these processes is the linear reaction-diffusion system \cite{Ellery2013,Simpson2015}
\begin{equation}
u_{\tau} = a_1 u_{\xi \xi} -b_1 u + v,  \label{LRD1}
\end{equation}
\begin{equation}
v_{\tau} = a_1 v_{\xi \xi} -b_2 v,  \label{LRD2}
\end{equation}
where the coefficients $a_1, \ b_i$ $(i = 1,2)$  are nonnegative. Another application is the study of the motion and biodegradation of dissolved organic contaminants in a saturated porous medium \cite{Simpson2007}. The following  explicit solutions are reported in the literature \cite{Chou2007,RajniRohila2016}:
\begin{equation}
\label{sol1a}
  u(\xi,\tau) = \left[ e^{-(b_1 + a_1)\tau} + e^{-(b_2 + a_1)\tau} \right]\cos \xi,  
\end{equation}
\begin{equation}
\label{sol1b}
  v(\xi,\tau) = (b_1 - b_2)e^{-(b_2 + a_1)\tau} \cos \xi.
\end{equation}
A family of exact solutions for (\ref{LRD1})-(\ref{LRD2}) were constructed in the framework of a growing domain (time-dependence domain) \cite{Simpson2015}. But, for the papers purpose,  we will only consider the solutions described above. 

\subsection{The diffusive Lotka-Volterra system}

The diffusive Lotka-Volterra (DLV) system is a generalization of the classical Lotka-Volterra system \cite{Lotka1920,Volterra1926}, in which the space diffusion is taking into account in order to have a more realistic population dynamics description. This system is used as model in different processes in biology \cite{Britton2003}, ecology \cite{Okubo2001}, and medicine \cite{Kuang2016}. 

The two-component DLV system 
\begin{equation}
   u_{\tau} = u_{\xi \xi} + u(a_1 -b_1 u - c_1 v),  \label{RD1}
\end{equation}
\begin{equation}
 v_{\tau} = v_{\xi \xi} + v(a_2 -b_2 u - c_2 v) \label{RD2}
\end{equation}
with $a_i, b_i, c_i\geq 0$ $(i = 1,2)$ describes competition between two populations $u$ and $v$, in which  diffusivity is the same for both populations. This DLV system admits the traveling wave solutions  \cite{Roman2022,Cherniha2004,Lu2022}:

\begin{equation}
\label{sol2a}
u(\xi, \tau) = \frac{A}{4B}\left[1 - \tanh \left(\frac{\sqrt{A}}{24}\xi - \frac{5A}{12}\tau \right)\right]^2,
\end{equation}
\begin{equation}
\label{sol2b}
 v(\xi, \tau) = \nu_0 + \nu_{1} \frac{A}{4B}\left[1 - \tanh \left(\frac{\sqrt{A}}{24}\xi - \frac{5A}{12}\tau \right)\right]^2.
\end{equation}
Here, the parameters $A, \ B, \ \nu_0, \ \nu_1$, are defined by the relations:
$$
A = \left\{ \begin{array}{lcc} a_1 = a_2, & if & \nu_0 = 0, \\ \\ a_1-a_2\frac{c_1}{c_2}, & if & \nu_0 = \frac{a_2}{c_2}  \end{array}\right., \quad  B = \left\{ \begin{array}{lcc} \frac{c_1 b_2 - b_1 c_2}{c_1 - c_2}, & if & \nu_0 = 0, \\ \\ b_1 + c_1 \nu_1, & if & \nu_0 = \frac{a_2}{c_2}  \end{array}\right.,
$$

$$
\nu_1 = \left\{ \begin{array}{lcc} \frac{b_1-b_2}{c_2-c_1}, & if & c_{1} \not = c_2, \ b_1 \not = b_2, \\ \\ -\frac{a_2 b_1}{a_1 c_1}, & if & c_1 = c_2, \ b_1 = b_2.  \end{array}\right. 
$$
It is well known that  these solutions possess the following properties  \cite{Roman2022,Cherniha2004,Lu2022}: If $\nu_0 \not = 0$, the solutions have the asymptotical behavior 
\begin{equation}
    (u,v) \rightarrow \left(\frac{a_1}{b_1},0\right),  \quad \quad \mbox{as} \quad  \quad \tau \rightarrow \infty, \label{Assypt_1}
\end{equation}
provided the condition $\hat{A} > \max \{ \hat{B}, \hat{C}\}$ where $\hat{A} = \frac{a_1}{a_2},$ $\hat{B} = \frac{b_1}{b_2},$ and $\hat{C} = \frac{c_1}{c_2}$ is satisfied. In contrast, if $\nu_0 = 0$ (with $a_1 = a_2$), then the solutions converge to the stationary state
\begin{equation}
    (u,v) \rightarrow \left(\frac{a_1(\hat{C}-1)}{b_2(\hat{C}-\hat{B})},\frac{a_1(1-\hat{B})}{c_2(\hat{C}-\hat{B})}\right), \quad  \quad \mbox{as} \quad \quad  \tau \rightarrow \infty. \label{Assypt_2}
\end{equation}
The asymptotic behavior is valid if one of these relations is satisfied: $\hat{B}> \hat{A} = 1 > \hat{C}$ or $\hat{C}> \hat{A} = 1 > \hat{B}$. In a biological context, the first property can describe the competition of two species, in which one of them dominates and the other eventually dies out. Behavior (\ref{Assypt_2}) represents an arbitrary long  coexistence of the species.  

In order to describe a larger number of population interaction  such as a predator-predator-prey model and a predator-prey-competition model, the three-component  DLV system is commonly assumed in the literature. This DLV system reads \cite{Roman2022,Pao2004,Pao2008}:
\begin{equation}
   u_{\tau} = u_{\xi \xi} + u(a_1 - b_1 u - c_1 v - e_1w),  \label{RD3}
\end{equation}
\begin{equation}
   v_{\tau} = v_{\xi \xi} + v(a_2 - b_2 u - c_2 v - e_2w),  \label{RD4}
\end{equation}
\begin{equation}
   w_{\tau} = w_{\xi \xi} + w(a_3 - b_3 u - c_3 v - e_3 w),\label{RD5}
\end{equation}
where $a_i, b_i, c_i, e_i,$ $(i = 1,2,3)$ are constants. As in the previous case, system (\ref{RD3})-(\ref{RD5}) possesses solutions in a  traveling wave form \cite{Roman2022,Li_Chang2011}:
\begin{equation}
\label{sol4a}
u(\xi, \tau) = \left(2 + \theta - \frac{a_1}{4}\right)\left[ 1- \tanh\left( \xi - \theta \tau \right)  \right]^2,
\end{equation}
\begin{equation}
\label{sol4b}
v(\xi, \tau) = \frac{a_1}{4}\left[ 1 + \tanh\left( \xi - \theta \tau \right)  \right]^2,
\end{equation}
\begin{equation}
\label{sol4c}
w(\xi, \tau) = \left(a_1 - \theta - 2\right)\left[ 1- \tanh\left( \xi - \theta \tau \right)  \right],
\end{equation}
as long as   $a_1 = a_2 = a_3,$ $b_1 = 1, \ b_2 = -\frac{a_1-24}{8-a_1 + 4 \theta}, \ b_3 = -\frac{a_1-4-2 \theta}{8-a_1+4 \theta}, $ $c_1 = -\frac{4\theta-a_1-16}{a_1}, \ c_2 = 1, \ c_3 = -\frac{2\theta-a_1-4}{a_1}, \ e_1 = -\frac{a_1-4-2\theta}{2+\theta-a_1}, \ e_2 = -\frac{a_1-4+2\theta}{2+\theta-a_1}, \ e_3 = 1, $ with $\theta + 2 < a_1 < 4(\theta +2)$. The last inequality is used to guarantee the positivity of the solutions.
\subsection{The Gray-Scott model}
The following is one the most important reaction-diffusion systems due to the diversity of phenomena that can be described by the model. More specifically, we consider the system of equations 
\begin{equation}
u_{\tau} = u_{\xi \xi} - uv^2 + b_1(1-u),  \label{NRD1}
\end{equation}
\begin{equation}
v_{\tau} = v_{\xi \xi} + uv^2 -b_1v,  \label{NRD2}
\end{equation}
where $b_1$ is the diffusion constant \cite{GRAY1984,Navneet2022,RODRIGO2001}. Gray and Scott \cite{GRAY1984} introduced this system in 1984 as a mathematical model for the cubic autocatalytic reaction. One of the exact solutions reported in the literature has the form \cite{Navneet2022,RODRIGO2001}  
\begin{equation}
\label{sol6a}
   u(\xi,\tau) = \frac{3-\sqrt{1-4b_1}}{4} - \frac{\sqrt{2 + 2\sqrt{1-4b_1}-4b_1}}{4}\tanh \left( \frac{\sqrt{1 + \sqrt{1-4b_1}-2b_1}}{4}(\xi-\theta \tau) \right),
\end{equation}
\begin{equation}
\label{sol6b}
   v(\xi,\tau) = \frac{1+\sqrt{1-4b_1}}{4} + \frac{\sqrt{2 + 2\sqrt{1-4b_1}-4b_1}}{4}\tanh \left( \frac{\sqrt{1 + \sqrt{1-4b_1}-2b_1}}{4}(\xi-\theta \tau) \right),
\end{equation}
with $\theta = \frac{\sqrt{2}(1-3\sqrt{1-4b_1})}{4}.$  Travelling waves solutions allow us to describe a large number of chemically
reacting systems \cite{Field1985,Kuramoto1984}. Recently in \cite{Wenjing2024}, numerical simulations of  pattern formation for more complex solutions revealed the potential applicability in 3D printing.   
\subsection{The coupled Burgers system}
The coupled Burgers model is a system of partial differential equations as follows:
\begin{equation}
  u_{\tau} = u_{\xi \xi}-b_{1}uu_{\xi}-c_{1}(uv)_{\xi}, \label{Bur1}  
\end{equation}
\begin{equation}
  v_{\tau} = v_{\xi \xi}-b_{2}vv_{\xi}-c_{2}(uv)_{\xi},   \label{Bur2}
\end{equation}
with $b_{i}, c_{i}$ $(i = 1,2)$ as  constants. The  system (\ref{Bur1})-(\ref{Bur2})  has been studied extensively from  theoretical and numerical point of view \cite{Abazari2010,Mehdi2007,Manoj2014,Mittal2011,Abdul2007}. As in the previous models, the coupled Burgers equations admit a travelling wave solution in the form \cite{Abazari2010,Manoj2014}:
\begin{equation}
\label{sol7a}
u(\xi,\tau) = B -2A(2c_1 - 1/4c_{1}c_{2} -1)\tanh \left[A\left(20\xi -10 -2A\tau \right)\right],
\end{equation}
\begin{equation}
\label{sol7b}
v(\xi,\tau) = B (2c_1 - 1/4c_{1}c_{2} -1)\tanh \left[A\left(20\xi -10 -2A\tau \right)\right],
\end{equation}
where $A = \left(4c_1c_2-1/4c_1 -2\right).$ Multiple-kink solutions are also reported for this system \cite{Jaradat2018}. In the field of fluid dynamics, the coupled Burgers equation plays a crucial role in the study of turbulence because it can be considered as a simple case of the Navier–Stokes equations.
\section{The Generalized Coupled Reaction-Diffusion and Burgers-type Systems}\label{Sect3}
This section demonstrates how to generate explicit solutions for general variable coefficients systems. Each of these general models can be viewed as extensions of the systems discussed in the preceding section. In our findings, the integrability of such systems is obtained on the premise of an equilibrium in the coefficients, which essentially assumes that they satisfy a Riccati system; as a result, the solutions of models with variable coefficients, which have more complicated dynamics, are obtained by transforming those solutions of models corresponding to systems of constant coefficients.
\subsection{Solutions for the generalized linear reaction-diffusion system}
Let's start with the first result of this section, which allows the construction of explicit solutions of the general linear reaction-diffusion model.  
\begin{theorem}[\textbf{Generalized Linear Reaction-Diffusion System}]\label{Th1}
The variable coefficients reaction-diffusion system 
\begin{eqnarray}
\psi _{t} &=&a(t)\psi_{xx} -(b\left(
t\right) x^{2}-d(t) - L_{1}(t)-xf(t))\psi -(g(t)-c\left( t\right) x)\psi
_{x}+ h(t)\varphi,  \label{GLRDS1} \\
\varphi _{t} &=&a(t)\varphi_{xx} -(b\left(
t\right) x^{2}-d(t) - L_{2}(t)-xf(t))\varphi -(g(t)-c\left( t\right) x)\varphi
_{x}  \label{GLRDS2}
\end{eqnarray}
can be transformed into the constant coefficients system 
\begin{equation}
  u_{\tau} = u_{\xi \xi}  -b_1 u + v,  \label{LRDS1}
\end{equation}
\begin{equation}
 v_{\tau} = v_{\xi \xi} -b_2 v.  \label{LRDS2}
\end{equation}
\end{theorem}

\begin{proof}

We will consider solutions for (\ref{GLRDS1})-(\ref{GLRDS2}) of the form 
\begin{equation}
\psi (x,t)=\frac{1}{\sqrt{\mu (t)}}e^{\alpha (t)x^{2}+\delta (t)x+\kappa
(t)}u(\xi ,\tau ),\qquad \xi =\beta (t)x+\varepsilon (t),\qquad \tau
=\gamma (t),\label{subst__1}
\end{equation}
and 
\begin{equation}
\varphi (x,t)=\frac{1}{\sqrt{\mu (t)}}e^{\alpha (t)x^{2}+\delta
(t)x+\kappa (t)}v(\xi ,\tau ).\qquad  \label{subst__2}
\end{equation}
Let $S = \alpha (t)x^{2}+\delta
(t)x+\kappa (t). $  Then, the derivatives of the function $\psi$ (similarly, derivatives of $\phi$ are obtained by changing the function $u$ by $v$) are 
\begin{equation}
    \psi_x = \mu^{-1/2}e^{S}\left[ u(2\alpha x + \delta) +  \beta u_\xi\right], \label{D1}
\end{equation}
\begin{equation}
    \psi_{xx} = \mu^{-1/2}e^{S}\left[ u(2\alpha x + \delta)^2 +  2\beta u_\xi (2 \alpha x + \delta) + u(2\alpha x + \delta) + \beta^2 u_{\xi \xi}\right], \label{D2}
\end{equation}
\begin{equation}
\psi_t = \mu^{-1/2}e^{S}\left[ -\frac{1}{2}\mu^{-1}\mu^{\prime}u + u(\alpha^\prime x^2 + \delta^\prime x + \kappa^\prime) + \gamma^\prime u_\tau + u_\xi (\beta^\prime x + \varepsilon^\prime) \right]. \label{D3}
\end{equation}
\noindent After substituting (\ref{D1})-(\ref{D3}) into (\ref{GLRDS1})-(\ref{GLRDS2}), we obtain the Riccati system   \cite{CorderoSoto2008,Escorcia,Suazo18,Suazo,Suazo13}: 
\begin{equation}
\dfrac{d\alpha }{dt}+b(t)= 2c(t)\alpha +4a(t)\alpha^{2},  \label{Rica_1}
\end{equation}%
\begin{equation}
\dfrac{d\beta }{dt}= (c(t)+4a(t)\alpha(t))\beta,  \label{Rica_2}
\end{equation}%
\begin{equation}
\dfrac{d\gamma }{dt}= a(t)\beta^{2}(t),
\label{Rica_3}
\end{equation}%
\begin{equation}
\dfrac{d\delta }{dt} + 2\alpha (t)g(t) = (c(t)+4a(t)\alpha(t))\delta + f(t),
\label{Rica_4}
\end{equation}%
\begin{equation}
\dfrac{d\varepsilon }{dt}=(2a(t)\delta(t)-g(t))\beta (t),  \label{Rica_5}
\end{equation}%
\begin{equation}
\dfrac{d\kappa }{dt}=a(t)\delta^{2}(t)-g(t)\delta (t).  \label{Rica_6}
\end{equation}%
\ Considering the  substitution\ 
\begin{equation}
\alpha = -\dfrac{\mu ^{\prime }(t)}{4a(t)\mu (t)}-\dfrac{d(t)}{2a(t)%
},  \label{Sus_1}
\end{equation}%
it follows that the Riccati equation (\ref{Rica_1}) becomes (characteristic equation)
\begin{equation}
\mu ^{\prime \prime }-\eta (t)\mu ^{\prime }-4\sigma (t)\mu =0,
\label{Carac__1}
\end{equation}%
with\ 
\begin{equation}
\eta (t)=\frac{a^{\prime }}{a}+2c-4d,\hspace{1cm}\sigma (t)=ab+cd-d^{2}+%
\frac{d}{2}\left( \frac{a^{\prime }}{a}-\frac{d^{\prime }}{d}\right) .
\end{equation}%
\noindent Further, if we choose the conditions 
$$ h(t)= a(t)\beta^{2}(t), \quad \quad L_{i}(t) = -b_{i}a(t)\beta^2(t), \quad i = 1,2, $$
the functions $u(\xi,\tau)$ and $v(\xi,\tau)$ will satisfy the system (\ref{LRDS1})-(\ref{LRDS2}). The Riccati system (\ref{Rica_1})-(\ref{Rica_6}) is solved explicitly in terms of the fundamental solution of the characteristic equation (\ref{Carac__1}), as can be found in the Appendix of this work (see equations (\ref{mu})-(\ref{kappa0})).
\end{proof}
According to Theorem \ref{Th1}, the interaction coefficient $h(t)$ determines how the dynamics of the solution $\psi$ rely on the function $\varphi$. Clearly, each coefficient in the equation has a substantial impact on the dynamics of the solutions, but we will only focus on the diffusion and interaction coefficients in this work. The examples below demonstrate how the interaction coefficient might influence the dynamics of solutions when it exhibits two distinct types of decay.
\subsubsection{\textbf{Constant diffusivity and an interaction coefficient with exponential decay}} 
Consider the general reaction-diffusion system
\begin{eqnarray}
    \psi _{t} &=& \psi_{xx} -\psi  - b_{1}e^{-2t}\psi -x\psi_x +e^{-2t}\varphi,  \label{Ex1a} \\
    \varphi _{t} &=& \varphi_{xx} -\varphi - b_{2}e^{-2t}\varphi -x\varphi_x.  
\label{Ex1b}
\end{eqnarray}
 By solving the Riccati system (\ref{Rica_1})-(\ref{Rica_6}) with the initial conditions $\alpha(0) = \delta(0) = \varepsilon(0) = \kappa(0) = \gamma(0) = 0, \ \beta(0) = 1, \ \mu(0) = 1, $ we get
\begin{eqnarray*}
 \alpha(t)= \delta(t) = \varepsilon(t) = \kappa(t) = 0, \quad \beta(t) = e^{-t}, \quad  \gamma(t)=  e^{-t}\sinh t, \quad \mu(t) = e^{2t}.
\end{eqnarray*}
 Then, in line with Theorem \ref{Th1}, the system (\ref{Ex1a})-(\ref{Ex1b}) admits a solution given by the equations (\ref{subst__1})-(\ref{subst__2}) with the functions $u,v$ satisfying the relations (\ref{sol1a})-(\ref{sol1b}) (with $a_1 = 1$), i.e.,
\begin{equation}
  \psi(x,t) =e^{-t} \cos\left(e^{-t}x\right)\left[e^{-(b_1 +1)e^{-t}\sinh t } + e^{-(b_2 +1)e^{-t}\sinh t } \right], 
\end{equation}
\begin{equation}
  \varphi(x,t) = (b_1 - b_2)e^{-t}\cos\left(e^{-t}x\right)e^{-(b_2 +1)e^{-t}\sinh t }.
\end{equation}
\begin{figure}[h!]
\centering
\subfigure[Profile of the solution $\psi$.]{\includegraphics[scale=0.32]{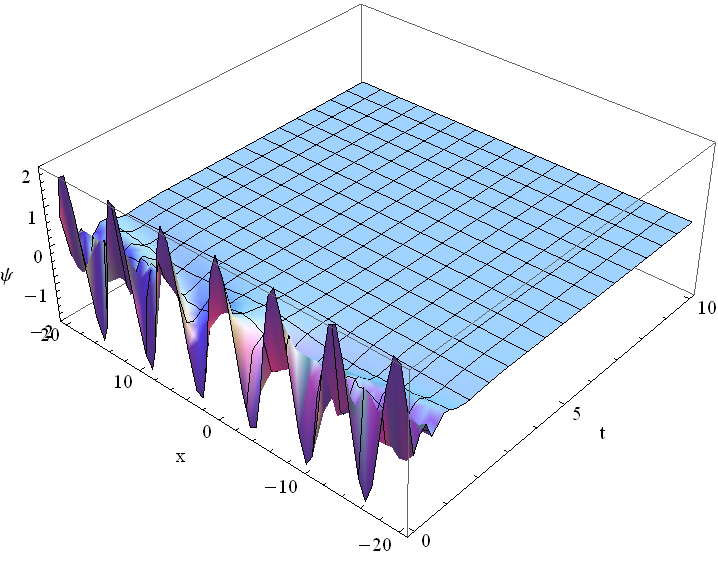}}
\subfigure[Contour of the solution $\psi$.]{\includegraphics[scale=0.28]{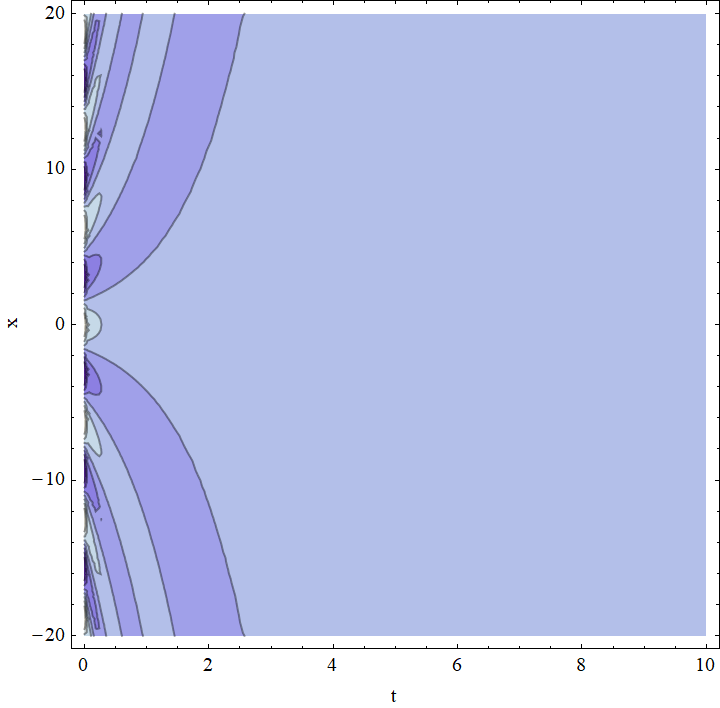}}
\subfigure[Profile of the solution $\varphi$.]{\includegraphics[scale=0.31]{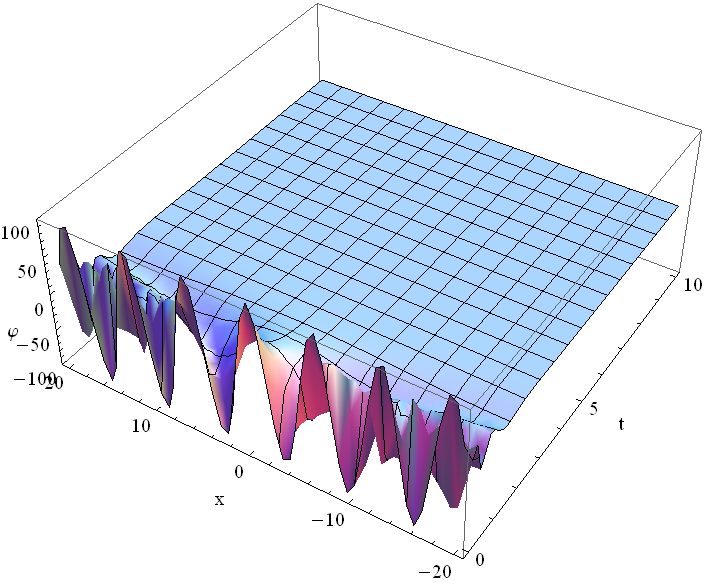}}
\subfigure[Contour of the solution $\varphi$.]{\includegraphics[scale=0.35]{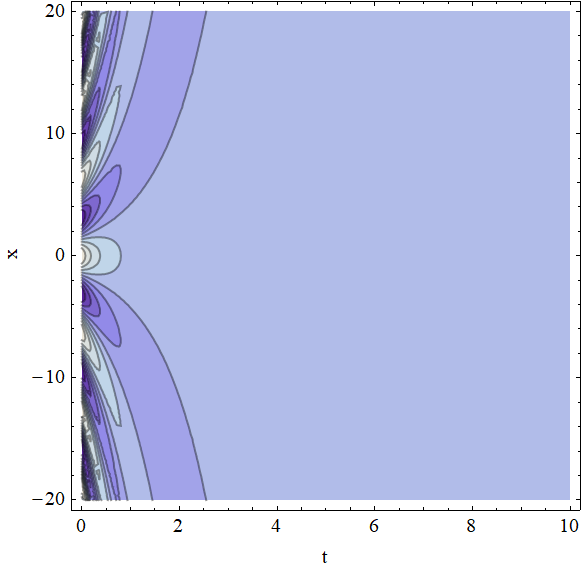}}
\caption{Solutions for the system  (\ref{Ex1a})-(\ref{Ex1b}) for the parameters $a_1 = 1,$ $b_1 = 100,$ and $b_2 = 1$. Here, (a) and (c) describe the profiles of the functions $\psi$ and $\varphi$ respectively. The corresponding contours of $\psi$ and $\varphi$ are shown in (b) and (d). }\label{Fig1}
\end{figure}
The time-evolution of these solutions for  $b_1 = 100$ and $b_2 = 1$, are shown in Figure \ref{Fig1}. Therefore we can observe the expected  exponential decay of the solutions. 
\subsubsection{\textbf{Periodic diffusivity and an interaction coefficient with rational decay}}
In this case, we take into account the following  system of equations:
{\footnotesize
\begin{eqnarray}
    \psi _{t} &=& e^{-2\cos t}\psi_{xx} + (e^{2\cos t}x^2 + 1)\psi  - \frac{b_1 e^2 \psi}{\left[e^2 + (16-2e^2)t\right]^2} + (2-\sin t)x\psi_x + \frac{e^2 \varphi}{\left[e^2 + (16-2e^2)t\right]^2},  \label{Ex2a} \\
    \varphi _{t} &=& e^{-2\cos t}\varphi_{xx} + (e^{2\cos t}x^2 + 1)\varphi  - \frac{b_2 e^2 \varphi}{\left[e^2 + (16-2e^2)t\right]^2} + (2-\sin t)x\varphi_x. 
\label{Ex2b}
\end{eqnarray}}
 The solution of the associated Riccati system is
\begin{eqnarray*}
 \alpha(t)= \frac{e^{2\cos t}\left[4 + (8-e^2)t\right]}{t(-16+2e^2)-e^2}, \quad \beta(t) &=&  \frac{e^{1+\cos t}}{t(16-2e^2)+e^2}, \quad  \gamma(t)= \frac{t}{t(16-2e^2)+e^2},
\end{eqnarray*}
\begin{figure}[h!]
\centering
\subfigure[Profile of the solution $\psi$.]{\includegraphics[scale=0.32]{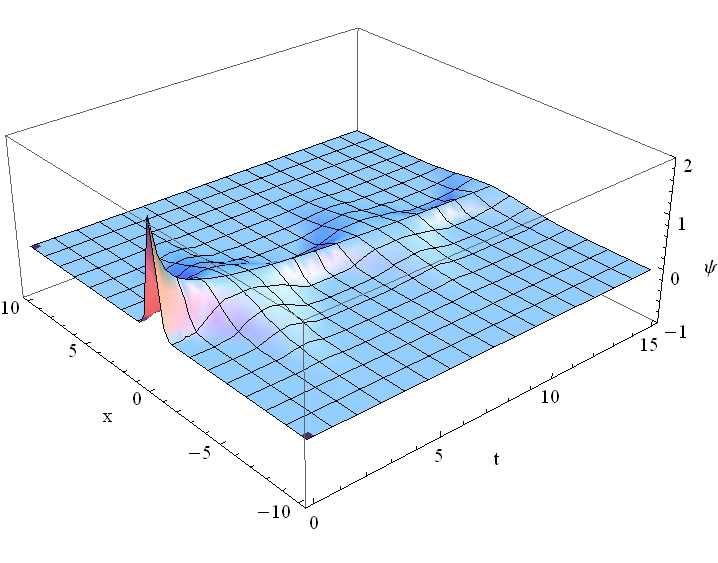}}
\subfigure[Contour of the solution $\psi$.]{\includegraphics[scale=0.25]{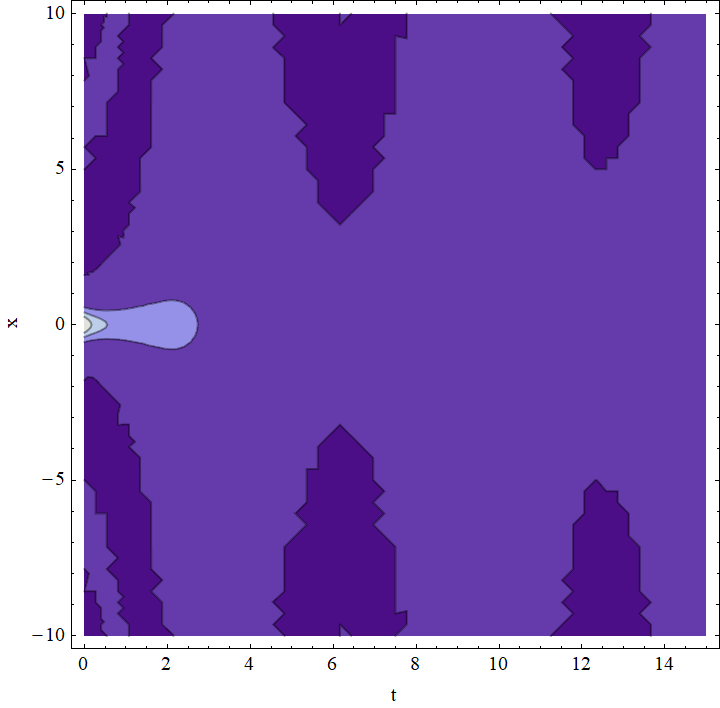}}
\subfigure[Profile of the solution $\varphi$.]{\includegraphics[scale=0.32]{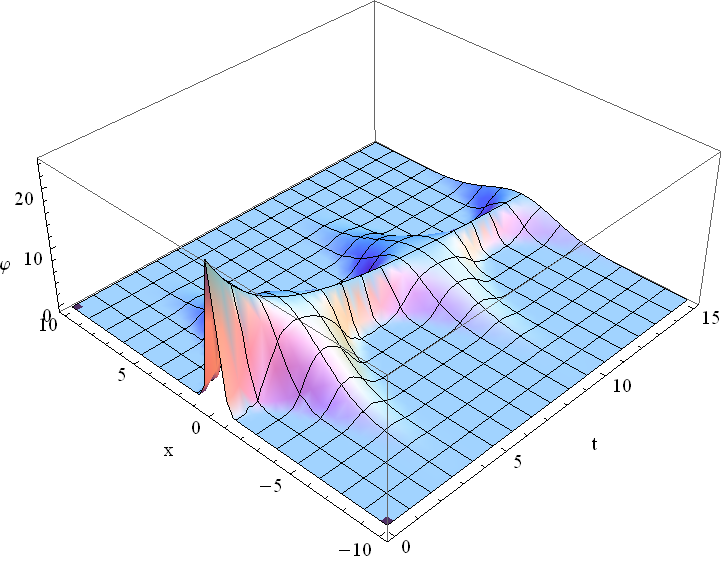}}
\subfigure[Contour of the solution $\varphi$.]{\includegraphics[scale=0.31]{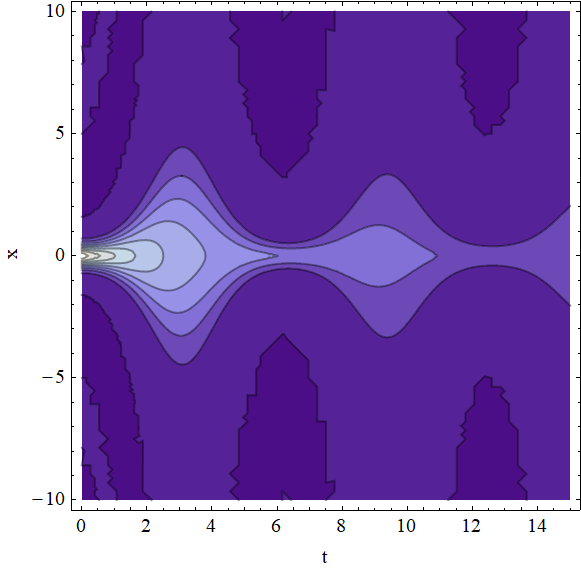}}
\caption{Solutions for the system  (\ref{Ex2a})-(\ref{Ex2b}) for the parameters $a_1 = 1,$ $b_1 = 25,$ and $b_2 = 1$. Here, (a) and (c) describe the profiles of the functions $\psi$ and $\varphi$ respectively. The corresponding contours of $\psi$ and $\varphi$ are shown in (b) and (d).}\label{Fig2}
\end{figure}

and
\begin{equation*}
\delta(t) \ = \ \varepsilon(t) = \kappa(t) = 0, \quad  \quad \mu(t) = 1 + t(16e^{-2}-2).
\end{equation*}
Here, we have assumed  the initial conditions $\alpha(0) = -4, \ \beta(0) = 1,\ \mu(0) = 1,\    \delta(0) = \varepsilon(0) =  \kappa(0) = \gamma(0) = 0.$ 
 Therefore, the system (\ref{Ex2a})-(\ref{Ex2b}) admits the explicit solutions 
{\footnotesize
\begin{equation}
  \psi(x,t) = \frac{\cos\left( \frac{e^{1+\cos t}}{t(16-2e^2)+e^2}x\right)}{\sqrt{1 + t(16e^{-2}-2)   }}\exp\left(\frac{e^{2\cos t}\left[4 + (8-e^2)t\right]}{t(-16+2e^2)-e^2}x^2\right)\left[e^{-(b_1 +1)t/t(16-2e^2)+e^2 } + e^{-(b_2 +1)t/t(16-2e^2)+e^2 } \right],  
\end{equation}}
\begin{equation}
  \varphi(x,t) = \frac{(b_1 - b_2)\cos\left( \frac{e^{1+\cos t}}{t(16-2e^2)+e^2}x\right)}{\sqrt{1 + t(16e^{-2}-2)   }}e^{-(b_2 +1)t/t(16-2e^2)+e^2 }.
\end{equation} 
The solutions  $\psi, \ \varphi$ are shown in Figure \ref{Fig2}.

The next result allows the construction of exponential-type solutions for the generalized linear reaction-diffusion equation. As we will see, the construction of these solutions requires the introduction of a \emph{modified Riccati system} which differs from (\ref{Rica_1})-(\ref{Rica_6}) because of the equations for the new functions $\kappa_1$ and $\kappa_2$.   
\begin{theorem}[\textbf{Exponential-type Solutions for the Linear Reaction-Diffusion System}]\label{Th1_1}
The variable coefficient coupled reaction-diffusion system 
\begin{eqnarray}
\psi _{t} &=&a(t)\psi_{xx} -(b\left(
t\right) x^{2}-d(t) -xf(t))\psi -(g(t)-c\left( t\right) x)\psi
_{x}+ h(t)\varphi,  \label{LinearRD1} \\
\varphi _{t} &=&a(t)\varphi_{xx} -(b\left(
t\right) x^{2}-d(t) -xf(t))\varphi -(g(t)-c\left( t\right) x)\varphi
_{x} \label{LinearRD2}
\end{eqnarray}
admits the explicit solutions
\begin{equation}
\psi (x,t)=\frac{1}{\sqrt{\mu (t)}}e^{\alpha (t)x^{2}+ \beta(t)xy + \gamma(t)y^2 +\delta
(t)x+\varepsilon(t)y +\kappa_{1}(t)}\label{substitutioN_1}
\end{equation}
and 
\begin{equation}
\varphi (x,t)=\frac{1}{\sqrt{\mu (t)}}e^{\alpha (t)x^{2}+ \beta(t)xy + \gamma(t)y^2 +\delta
(t)x+\varepsilon(t)y +\kappa_{1}(t)+\kappa_{2}(t)}, \label{substitutioN_2}
\end{equation}
with $y$ as a real parameter and the functions $\alpha(t), \  \beta(t), \ \gamma(t), \ \delta(t), \ \varepsilon(t), \ \kappa_1(t), \ \kappa_{2}(t), \ \mu(t),$ satisfying the modified Riccati system (\ref{Ricati1})-(\ref{Ricati7}).
\end{theorem}

\begin{proof}
Let $S_i = \alpha (t)x^{2}+ \beta(t)xy + \gamma(t)y^2 +\delta
(t)x+\varepsilon(t)y +\kappa_{1}(t) + (i-1)\kappa_{2}(t)$ for $i = 1,2$.  In these terms, the derivatives of the function $\psi$ are as follows (similarly, derivatives of $\phi$ are obtained by changing $S_1$ for $S_2$ and adding the term  $\kappa_2^\prime$, in the temporal derivative):
\begin{equation}
    \psi_x = \mu^{-1/2}e^{S_1}\left[2\alpha x + \beta y + \delta \right], \label{DD1}
\end{equation}
\begin{equation}
      \psi_{xx} = \mu^{-1/2}e^{S_1}\left[ (2\alpha x + \beta y + \delta)^2 + 2\alpha\right], \label{DD2}
\end{equation}
\begin{equation}
 \psi_t = \mu^{-1/2}e^{S_1}\left[ -\frac{1}{2}\mu^{-1}\mu^{\prime} + \alpha^\prime x^2 + \beta^\prime xy + \gamma^\prime y^2 + \delta^\prime x + \varepsilon^\prime y +  \kappa_1 ^\prime \right]. \label{DD3}
\end{equation}
Now, substituting (\ref{DD1})-(\ref{DD3}) into the equations  (\ref{LinearRD1})-(\ref{LinearRD2}) we get the \emph{modified Riccati system}:
\begin{equation}
\dfrac{d\alpha }{dt}+b(t)= 2c(t)\alpha +4a(t)\alpha^{2},  \label{Ricati1}
\end{equation}%
\begin{equation}
\dfrac{d\beta }{dt}= (c(t)+4a(t)\alpha(t))\beta,  \label{Ricati2}
\end{equation}%
\begin{equation}
\dfrac{d\gamma }{dt}= a(t)\beta^{2}(t),
\label{Ricati3}
\end{equation}%
\begin{equation}
\dfrac{d\delta }{dt} + 2\alpha (t)g(t) = (c(t)+4a(t)\alpha(t))\delta + f(t),
\label{Ricati4}
\end{equation}%
\begin{equation}
\dfrac{d\varepsilon }{dt}=(2a(t)\delta(t)-g(t))\beta (t),  \label{Ricati5}
\end{equation}%
\begin{equation}
\dfrac{d\kappa_1 }{dt}=a(t)\delta^{2}(t)-g(t)\delta (t) + h(t)e^{\kappa_2(t)},  \label{Ricati6}
\end{equation}%
\begin{equation}
\dfrac{d\kappa_2 }{dt}= -h(t)e^{\kappa_2(t)}. \label{Ricati7}
\end{equation}
Now, as usual, the  substitution
\begin{equation}
\alpha = -\dfrac{\mu ^{\prime }(t)}{4a(t)\mu (t)}-\dfrac{d(t)}{2a(t)%
}  \label{Ricati8}
\end{equation}%
transforms the Riccati equation (\ref{Ricati1}) into the characteristic equation (\ref{Carac__1}). In this sense, using equations (\ref{mu})-(\ref{kappa0}) we can solve (\ref{Ricati1})-(\ref{Ricati5}), but equations (\ref{Ricati6})-(\ref{Ricati7}) must be solved separately. In fact, the solution of this Riccati system can be found in the Appendix of this work.
\end{proof}
Unlike Theorem \ref{Th1}, this last result enables us to create a family of exponential-type solutions for the system (\ref{LinearRD1})-(\ref{LinearRD2}) in which both solutions, $\psi$ and $\varphi$, are directly dependent on the interaction coefficient $h(t)$. Next, an example of Theorem \ref{Th1_1} is shown. 
\subsubsection{\textbf{Constant diffusivity and a periodic interaction coefficient}}
The following  system with variable coefficients
{\small
\begin{eqnarray}
    \psi _{t} &=& \frac{1}{2}\psi_{xx} + 2\left(x^2 + \frac{1}{4}-3x\sin(3t)-x\cos(3t)\right)\psi  + \left(2x-\cos(3t) \right) \psi_x + 3\sin(6t)e^{\sin^2(3t)}\varphi,  \label{Ex2aa} \\
    \varphi _{t} &=& \frac{1}{2}\varphi_{xx} + 2\left(x^2 + \frac{1}{4}-3x\sin(3t)-x\cos(3t)\right)\varphi  + \left(2x-\cos(3t) \right) \varphi_x
\label{Ex2bb}
\end{eqnarray}}
has the  solutions 
\begin{equation}
  \psi(x,t) = \exp\left[-\left(x-\cos(3t)\right)^2-\frac{t}{2}\right],  
\end{equation}
\begin{equation}
  \varphi(x,t) = \exp\left[-\left(x-\cos(3t)\right)^2-\sin^2(3t)-\frac{t}{2}\right].
\end{equation} 
\begin{figure}[h!]
\centering
\subfigure[Profile of the solution $\psi$.]{\includegraphics[scale=0.32]{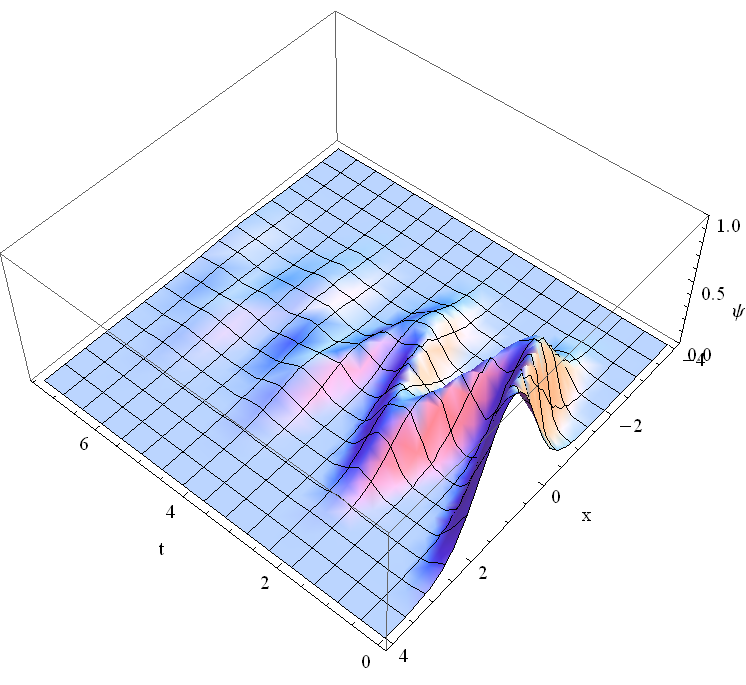}}
\subfigure[Contour of the solution $\psi$.]{\includegraphics[scale=0.27]{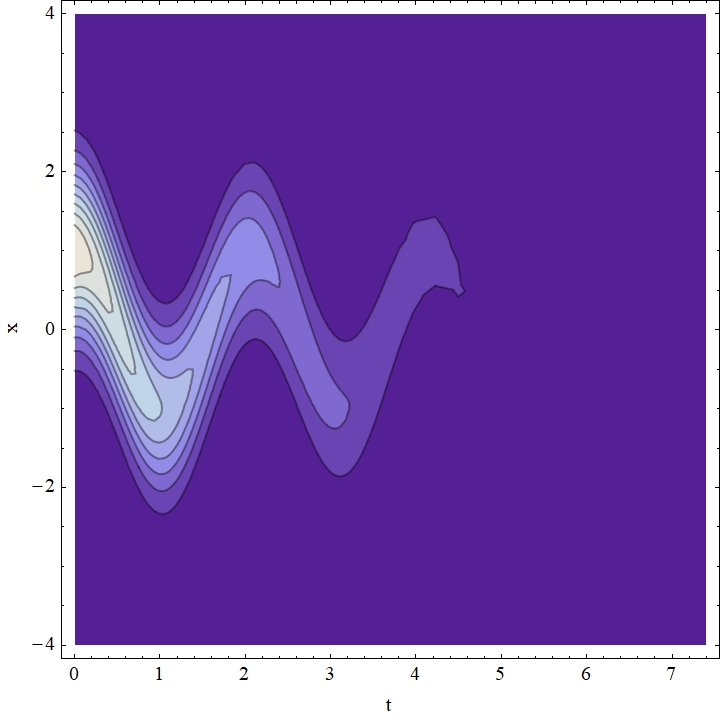}}
\subfigure[Profile of the solution $\varphi$.]{\includegraphics[scale=0.32]{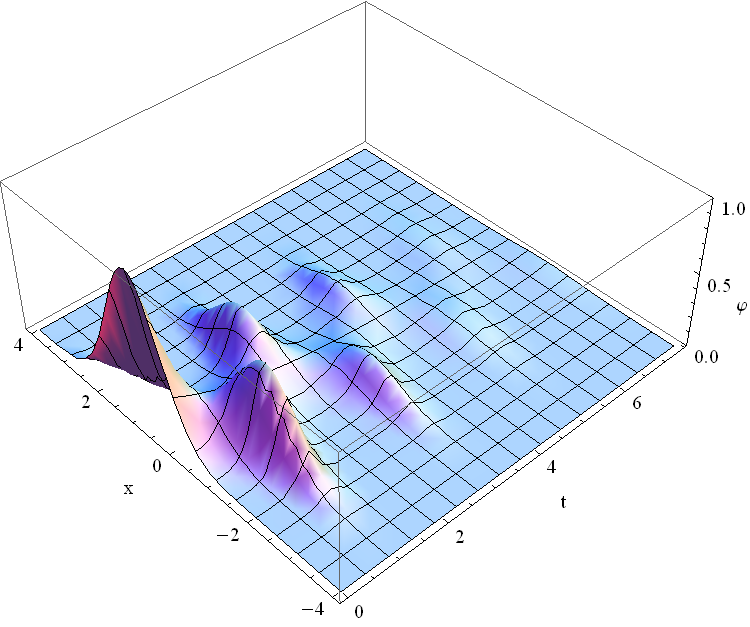}}
\subfigure[Contour of the solution $\varphi$.]{\includegraphics[scale=0.33]{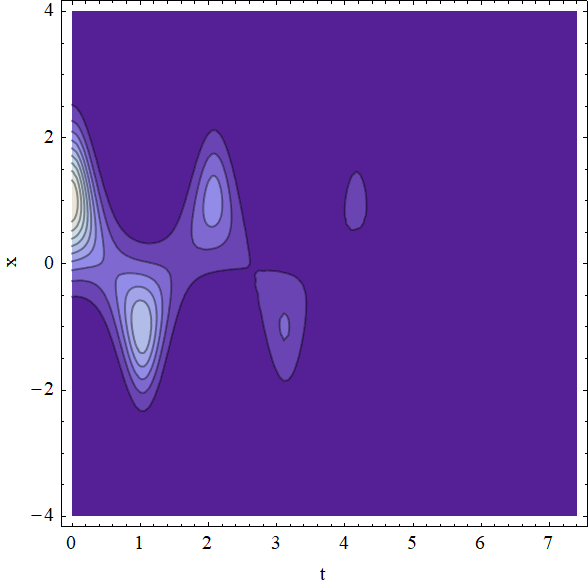}}
\caption{Solutions for the system  (\ref{Ex2aa})-(\ref{Ex2bb}). The  profiles of the functions $\psi$ and $\varphi$ are shown in (a) and (c). In the contours of $\psi$ and $\varphi$, Figures (b) and (d), the bending dynamics are clearly observed.}\label{Fig2_2}
\end{figure}
The periodicity of the interaction coefficient causes an interesting dynamic: The central axis of the solution  bends left and right over time and eventually disappears due to time exponential decay. The aforementioned dynamics can be seen in Figure \ref{Fig2_2}.
\subsection{Solutions for the generalized Lotka-Volterra system}
In the present section, we introduce a general diffusive Lotka-Volterra system (two and three components) and provide a mechanism for producing solutions by means of the classical models (\ref{RD1})-(\ref{RD2}) and (\ref{RD3})-(\ref{RD5}).
\begin{theorem}[\textbf{Generalized Diffusive Lotka-Volterra System}]\label{Th2}
The variable coefficient coupled reaction-diffusion system 
{\footnotesize
\begin{eqnarray}
\psi _{t} &=&a(t)\psi_{xx} -(b\left(
t\right) x^{2}-d(t) - L_{1}(t)-xf(t))\psi -(g(t)-c\left( t\right) x)\psi
_{x}+ (h_{1}(x,t)\psi + r_1(x,t)\varphi)\psi,  \label{GRDS1} \\
\varphi _{t} &=&a(t)\varphi_{xx} -(b\left(
t\right) x^{2}-d(t) - L_{2}(t)-xf(t))\varphi -(g(t)-c\left( t\right) x)\varphi
_{x}+ (h_{2}(x,t)\psi + r_2(x,t) \varphi)\varphi,  \label{GRDS2}
\end{eqnarray}%
} can be transformed into the classical Lotka-Volterra system
\begin{equation}
  u_{\tau} = u_{\xi \xi} + u(a_1 -b_1 u - c_1 v),  \label{RDS1}
\end{equation}
\begin{equation}
 v_{\tau} = v_{\xi \xi} + v(a_2 -b_2 u - c_2 v).  \label{RDS2}
\end{equation}
\end{theorem}

\begin{proof}
We will assume solutions for (\ref{GRDS1})-(\ref{GRDS2}) in the form 
\begin{equation}
\psi (x,t)=\frac{1}{\sqrt{\mu (t)}}e^{\alpha (t)x^{2}+\delta (t)x+\kappa
(t)}u(\xi ,\tau ),\qquad \xi =\beta (t)x+\varepsilon (t),\qquad \tau
=\gamma (t),\label{substitution_1}
\end{equation}
and 
\begin{equation}
\varphi (x,t)=\frac{1}{\sqrt{\mu (t)}}e^{\alpha (t)x^{2}+\delta
(t)x+\kappa (t)}v(\xi ,\tau ).\qquad  \label{substitution_2}
\end{equation}
\noindent After substituting (\ref{substitution_1})-(\ref{substitution_2}) in (\ref{GRDS1})-(\ref{GRDS2}), we obtain again the  Riccati system  (\ref{Rica_1})-(\ref{Rica_6}), i.e.,
\begin{equation}
\dfrac{d\alpha }{dt}+b(t)= 2c(t)\alpha +4a(t)\alpha^{2},  \label{rica1}
\end{equation}%
\begin{equation}
\dfrac{d\beta }{dt}= (c(t)+4a(t)\alpha(t))\beta,  \label{rica2}
\end{equation}%
\begin{equation}
\dfrac{d\gamma }{dt}= a(t)\beta^{2}(t),
\label{rica3}
\end{equation}%
\begin{equation}
\dfrac{d\delta }{dt} + 2\alpha (t)g(t) = (c(t)+4a(t)\alpha(t))\delta + f(t),
\label{rica4}
\end{equation}%
\begin{equation}
\dfrac{d\varepsilon }{dt}=(2a(t)\delta(t)-g(t))\beta (t),  \label{rica5}
\end{equation}%
\begin{equation}
\dfrac{d\kappa }{dt}=a(t)\delta^{2}(t)-g(t)\delta (t).  \label{rica6}
\end{equation}%
 Here, the characteristic equation (\ref{Carac__1}) can be obtained by the usual substitution (\ref{Ricati8}). Additionally, if we assume the conditions 
$$L_{i}(t) = a_{i}a(t)\beta^2(t), \quad \quad  h_{i}(x,t)=-b_{i} a(t)\beta^{2}(t)\mu^{1/2}(t)e^{-(\alpha (t)x^{2}+\delta
(t)x+\kappa (t))}, $$
$$  r_{i}(x,t)=-c_{i} a(t)\beta^{2}(t)\mu^{1/2}(t)e^{-(\alpha (t)x^{2}+\delta
(t)x+\kappa (t))}, \quad i = 1,2, $$
 the functions $u(\xi,\tau)$ and $v(\xi,\tau)$ will satisfy the system (\ref{RDS1})-(\ref{RDS2}). 
\end{proof}
In simple terms, the integrability of the variable coefficients Lotka-Volterra system is consequences of the Riccati system and the exponential structure of the interaction coefficients $h_i(x,t)$ and $r_i(x,t)$. The subsequent  examples illustrate the preceding result.
\subsubsection{\textbf{Constant diffusivity and tanh-type interaction coefficients}} In this example we demonstrate the existence of travelling wave-type solutions for the general model. Indeed, consider the system
{\small
\begin{eqnarray}
\psi _{t} &=&\frac{1}{2}\psi_{xx} + \frac{1}{2}\left(\tanh t + 2a_1 -1\right)\psi - \left(\frac{b_1 e^{-\kappa(0)}}{2\sqrt{1 + \tanh t}}\psi + \frac{c_1 e^{-\kappa(0)}}{2\sqrt{1 + \tanh t}}\varphi\right)\psi,  \label{Ex3a} \\
\varphi _{t} &=&\frac{1}{2}\varphi_{xx} + \frac{1}{2}\left(\tanh t + 2a_2 -1\right)\varphi - \left(\frac{b_2 e^{-\kappa(0)}}{2\sqrt{1 + \tanh t}}\psi + \frac{c_2 e^{-\kappa(0)}}{2\sqrt{1 + \tanh t}}\varphi\right)\varphi. \label{Ex3b}
\end{eqnarray}%
}
The solution of the Riccati system associated with this problem is
\begin{eqnarray*}
 \alpha(t)= \delta(t) = 0, \quad \beta(t) =  \sqrt{2}, \quad  \gamma(t)= t,  \quad  \varepsilon(t) = \varepsilon(0), \quad   \kappa(t) = \kappa(0), \quad \mu(t) = 1 + \tanh t.
\end{eqnarray*}
 Then, a solution for the system (\ref{Ex3a})-(\ref{Ex3b}) has
 the following traveling wave form:
\begin{equation}
  \psi(x,t) = \frac{Ae^{\kappa(0)}}{4B\sqrt{1 +\tanh t}} \left[1 - \tanh \left(\frac{\sqrt{A}}{24}(\sqrt{2}x+\varepsilon(0)) - \frac{5A}{12}t \right)\right]^2, 
\end{equation}
\begin{equation}
  \varphi(x,t) = \frac{e^{\kappa(0)}}{\sqrt{1 +\tanh t}} \left \{ \nu_0 + \nu_1\frac{A}{4B}\left[1 - \tanh \left(\frac{\sqrt{A}}{24}(\sqrt{2}x+\varepsilon(0)) - \frac{5A}{12}t \right)\right]^2 \right \}. 
\end{equation}
These solutions exhibit the asymptotic behaviors:
\begin{equation}
 \left(\psi,\varphi\right) \rightarrow \left(\frac{a_1 e^{\kappa(0)}}{\sqrt{2}b_1},0\right), \quad \mbox{as} \quad t \rightarrow \infty,   
\end{equation}

\begin{figure}[h!]
\centering
\subfigure[Profile of the solution $\psi$.]{\includegraphics[scale=0.31]{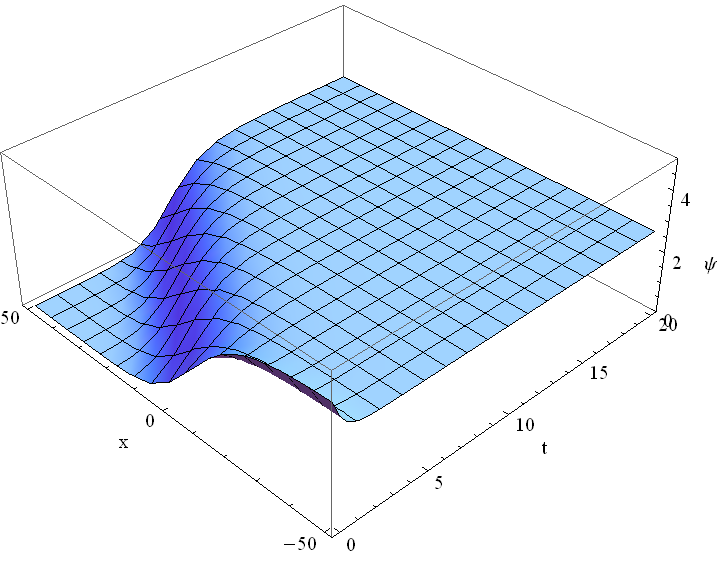}}
\subfigure[Contour of the solution $\psi$.]{\includegraphics[scale=0.31]{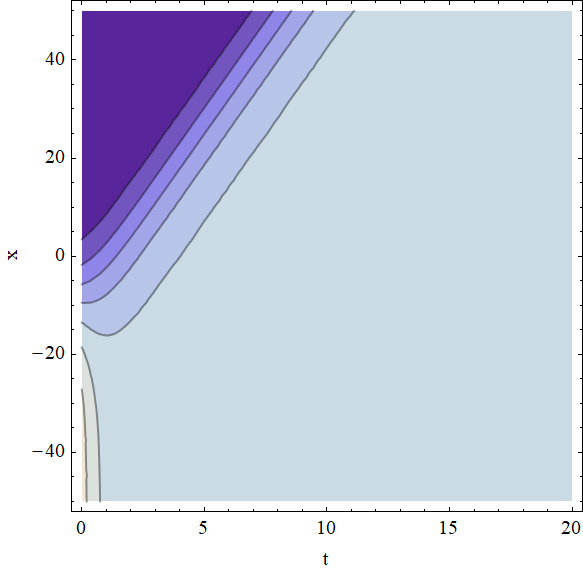}}
\subfigure[Profile of the solution $\varphi$.]{\includegraphics[scale=0.31]{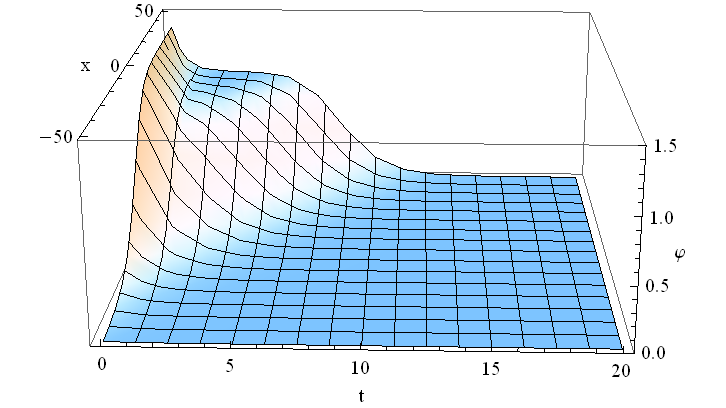}}
\subfigure[Contour of the solution $\varphi$.]{\includegraphics[scale=0.26]{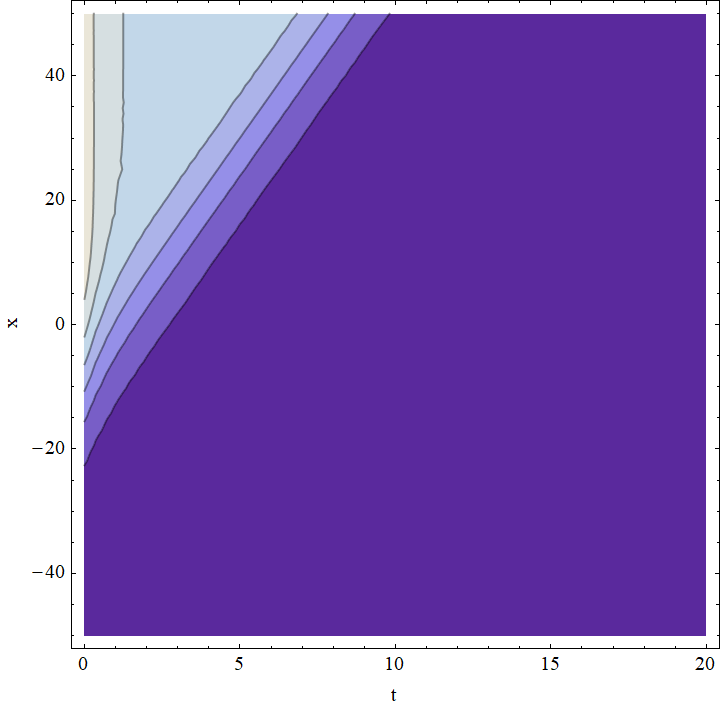}}
\caption{Figures (a) and (c) describe the traveling wave solutions for the system  (\ref{Ex3a})-(\ref{Ex3b}) with $a_1 = c_2 = c_1 = 2,$ $a_2 = 1,$ $b_1 = b_2 = \sqrt{2}$, $\nu_0 = \frac{1}{2},$ $\nu_1 = -\frac{\sqrt{2}}{4}$, $A = 1,$ $B = \frac{\sqrt{2}}{2}$, $\kappa(0) = 1$, and $\varepsilon(0) = 2$. The contours of these solutions are shown in (b) and (d). } \label{Fig3}
\end{figure}
or
\begin{equation}
   \left(\psi,\varphi\right) \rightarrow \left(\frac{a_1(\hat{C}-1)e^{\kappa(0)}}{\sqrt{2}b_2(\hat{C}-\hat{B})},\frac{a_1(1-\hat{B})e^{\kappa(0)}}{\sqrt{2}c_2(\hat{C}-\hat{B})}\right), \quad \mbox{as} \quad t \rightarrow \infty.
\end{equation}
The constants $\hat{B}, \hat{C},$ are the same as those established in Section \ref{Sect2}. The profiles of these solutions are displayed in Figure \ref{Fig3}. A \emph{plausible interpretation} of these solutions would be to depict two types of scenarios in a competition between two species. In the first scenario, one species dominates the other until it is extinct, whereas in the second situation, both species coexist over time. This is clear evidence that some aspects of the classical Lotka-Volterra equations are applicable to the general system.

\subsubsection{\textbf{Constant diffusivity and interaction coefficients with periodic central axes}}
Let's assume the system of variable coefficients equations
{\scriptsize
\begin{eqnarray}
\psi _{t} &=&\frac{1}{2}\psi_{xx} + \left(2 + \frac{1}{2}a_1e^{-4t}\right)\psi +\left(2x-\sin t -\frac{1}{2}\cos t\right)\psi_x  -\frac{1}{2}e^{-4t}\exp\left[2\left(x-\frac{1}{2}\sin t\right)^2\right] \left(b_1 \psi + c_1 \varphi\right)\psi,  \label{Ex4a} \\
\varphi _{t} &=&\frac{1}{2}\varphi_{xx} + \left(2 + \frac{1}{2}a_2e^{-4t}\right)\varphi +\left(2x-\sin t -\frac{1}{2}\cos t\right)\varphi_x  -\frac{1}{2}e^{-4t}\exp\left[2\left(x-\frac{1}{2}\sin t\right)^2\right] \left(b_2 \psi + c_2 \varphi\right)\varphi. \label{Ex4b}
\end{eqnarray}%
}
 In this case, we have the functions
\begin{eqnarray*}
 \alpha(t)= -2,  \quad \quad \beta(t) =  e^{-2t}, \quad  \quad \gamma(t)= \frac{1}{8}-\frac{e^{-4t}}{8}, \quad  \quad \delta(t) = 2\sin t, \\\ \varepsilon(t) = -\frac{1}{2}e^{-2t}\sin t, \quad \quad   \kappa(t) = -\frac{1}{2}\sin^2 t, \quad \quad  \mu(t) = 1,
\end{eqnarray*}
 and the system admits a solution given by the following  equations
\begin{equation}
  \psi(x,t) = \exp\left[-2\left(x-\frac{1}{2}\sin t\right)^2\right] \left[1 - \tanh \left(\frac{\sqrt{A}}{24}e^{-2t}(x-\frac{1}{2}\sin t) - \frac{5A}{96}(1-e^{-4t}) \right)\right]^2, 
\end{equation}
\begin{equation}
  \varphi(x,t) = \exp\left[-2\left(x-\frac{1}{2}\sin t\right)^2\right]  \left \{ \nu_0 + \nu_1\frac{A}{4B}\left[1 - \tanh \left(\frac{\sqrt{A}}{24}e^{-2t}(x-\frac{1}{2}\sin t) - \frac{5A}{96}(1-e^{-4t})  \right)\right]^2 \right \}. 
\end{equation}
As can be seen in Figure \ref{Fig4}, solutions share an interesting bending dynamic controlled by the function $\sin t$. Due to the similar behavior of these solutions and with the intention of providing the reader with visual proof of the difference between them, we report in Figure \ref{Fig4}(d) the profile of $|\psi - \varphi|$. 

The next proposition establishes a natural generalization of Theorem \ref{Th2}, corresponding to the three-component Lotka-Volterra equations. 
\begin{figure}[h!]
\centering
\subfigure[Profile of the solution $\psi$.]{\includegraphics[scale=0.34]{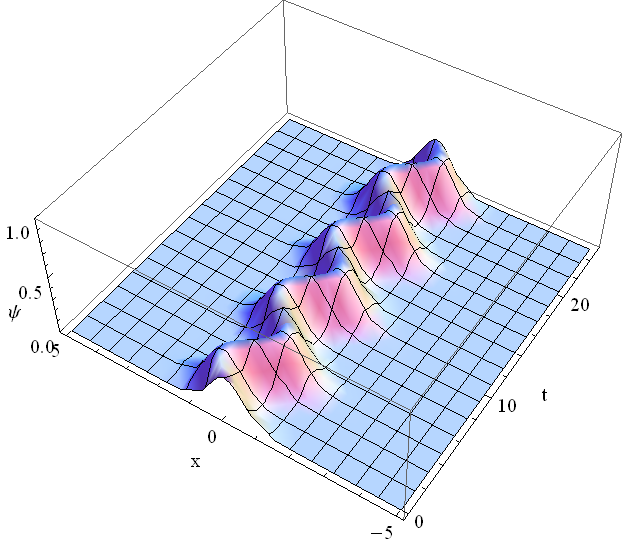}}
\subfigure[Contour of the solution $\psi$.]{\includegraphics[scale=0.31]{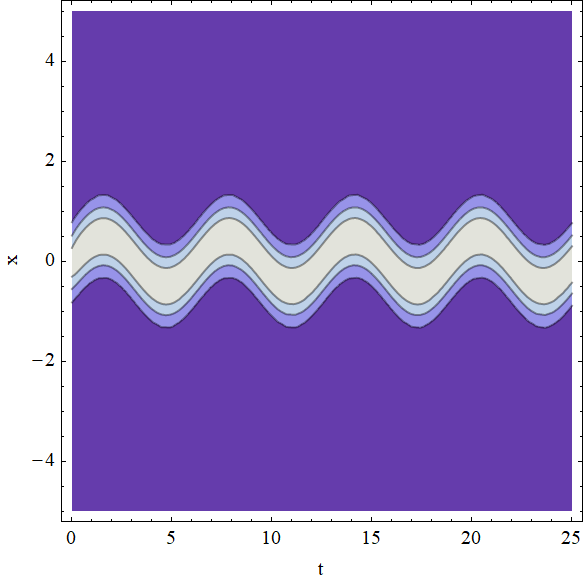}}
\subfigure[Profile of the solution $\varphi$.]{\includegraphics[scale=0.33]{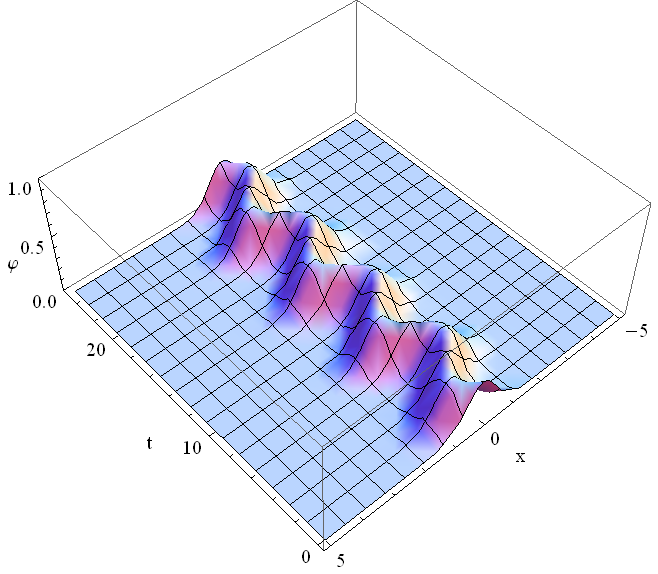}}
\subfigure[Profile of the difference $|\varphi-\psi|$.]{\includegraphics[scale=0.27]{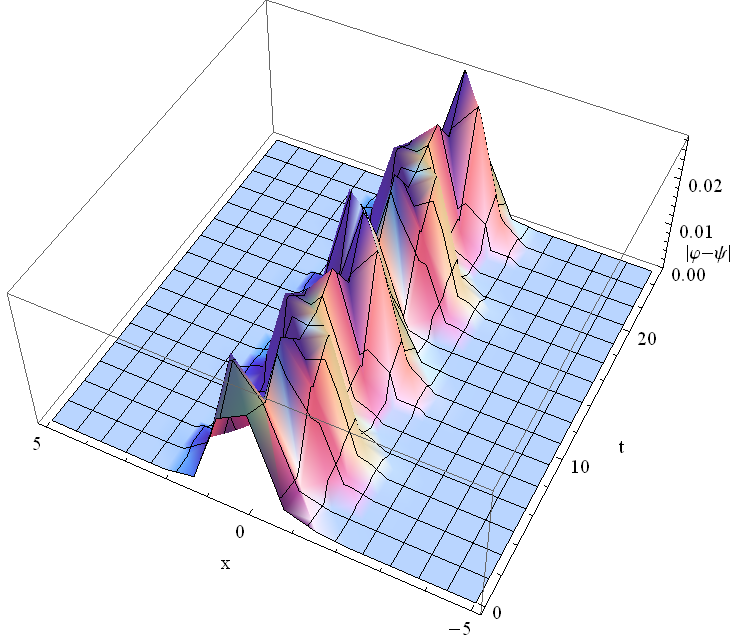}}
\caption{Solutions for the system  (\ref{Ex4a})-(\ref{Ex4b}) for the parameters $a_1 = c_2 = c_1 = 2,$ $a_2 = 1,$ $b_1 = b_2 = \sqrt{2}$, $\nu_0 = \frac{1}{2},$ $\nu_1 = -\frac{\sqrt{2}}{4}$, $A = 1,$ $B = \frac{\sqrt{2}}{2}$.} \label{Fig4}
\end{figure}

\begin{proposition}[\textbf{Generalized Three-Component Lotka-Volterra System}]\label{PROPOSITION1}
The general three-component Lotka-Volterra equations 
{\small  
\begin{eqnarray}
\psi _{t} &=&a(t)\psi_{xx} -(b\left(
t\right) x^{2}-d(t) - L_{1}(t)-xf(t))\psi -(g(t)-c\left( t\right) x)\psi
_{x}+ (h_{1}\psi + r_1\varphi + s_1\phi )\psi,  \label{GRDS_1} \\
\varphi _{t} &=&a(t)\varphi_{xx} -(b\left(
t\right) x^{2}-d(t) - L_{2}(t)-xf(t))\varphi -(g(t)-c\left( t\right) x)\varphi
_{x}+ (h_{2}\psi + r_2 \varphi + s_2\phi )\varphi,  \label{GRDS_2} \\
\phi _{t} &=&a(t)\phi_{xx} -(b\left(
t\right) x^{2}-d(t) - L_{3}(t)-xf(t))\phi -(g(t)-c\left( t\right) x)\phi
_{x}+ (h_{3}\psi + r_3 \varphi + s_3\phi )\phi  \label{GRDS_3}
\end{eqnarray}%
} can be transformed into the  system 
\begin{equation}
  u_{\tau} = u_{\xi \xi} + u(a_1 -b_1 u - c_1 v - e_1 w),  \label{RDS_1}
\end{equation}
\begin{equation}
 v_{\tau} = v_{\xi \xi} + v(a_2 -b_2 u - c_2 v - e_2w).  \label{RDS_2}
\end{equation}
\begin{equation}
 w_{\tau} = w_{\xi \xi} + w(a_3 -b_3 u - c_3 v - e_3w).  \label{RDS_3}
\end{equation}
\end{proposition}

\begin{proof}
By extending the same ideas from Theorem \ref{Th2}, we use the substitutions
\begin{equation}
\psi (x,t)=\frac{1}{\sqrt{\mu (t)}}e^{\alpha (t)x^{2}+\delta (t)x+\kappa
(t)}u(\xi ,\tau ),\qquad \xi =\beta (t)x+\varepsilon (t),\qquad \tau
=\gamma (t)\label{subst1}
\end{equation}
and 
\begin{equation}
\varphi (x,t)=\frac{1}{\sqrt{\mu (t)}}e^{\alpha (t)x^{2}+\delta
(t)x+\kappa (t)}v(\xi ,\tau ), \qquad \quad \phi (x,t)=\frac{1}{\sqrt{\mu (t)}}e^{\alpha (t)x^{2}+\delta
(t)x+\kappa (t)}w(\xi ,\tau ). \label{subst2}
\end{equation}
As a consequence of imposing the integrability conditions 
$$  r_{i}(x,t)=-c_{i} a(t)\beta^{2}(t)\mu^{1/2}(t)e^{-(\alpha (t)x^{2}+\delta
(t)x+\kappa (t))}, \quad  s_{i}(x,t)=-e_{i} a(t)\beta^{2}(t)\mu^{1/2}(t)e^{-(\alpha (t)x^{2}+\delta
(t)x+\kappa (t))}, $$
$$L_{i}(t) = a_{i}a(t)\beta^2(t), \quad \quad  h_{i}(x,t)=-b_{i} a(t)\beta^{2}(t)\mu^{1/2}(t)e^{-(\alpha (t)x^{2}+\delta
(t)x+\kappa (t))},  \quad i = 1,2,3,$$
and the coefficients to satisfy the Riccati system (\ref{rica1})-(\ref{rica6}) (and therefore equations (\ref{Sus_1})-(\ref{Carac__1})), the functions $u(\xi,\tau), \ v(\xi,\tau), \ w(\xi,\tau),$ will satisfy the system (\ref{RDS_1})-(\ref{RDS_3}). 
\end{proof}

\subsubsection{\textbf{Constant diffusivity and interaction coefficients with exponential growth}}
To exemplify this case, let us consider the system of equations
\begin{eqnarray}
\psi _{t} &=& \beta^{-2}(0)\psi_{xx} + \left(a_1 - 1 -\sech t \tanh t \right)\psi -\mu^{1/2}(0)e^{t-\sech t -\kappa(0)}\left(\psi + c_1 \varphi + e_1 \phi \right)\psi,  \label{Ex4_1a} \\
\varphi _{t} &=& \beta^{-2}(0)\varphi_{xx} + \left(a_1 - 1 -\sech t \tanh t \right)\varphi -\mu^{1/2}(0)e^{t-\sech t -\kappa(0)}\left(b_2 \psi +  \varphi + e_2 \phi \right)\varphi,  \label{Ex4_1b} \\
\phi_{t} &=& \beta^{-2}(0)\phi_{xx} + \left(a_1 - 1 -\sech t \tanh t \right)\phi -\mu^{1/2}(0)e^{t-\sech t -\kappa(0)}\left(b_3 \psi + c_3 \varphi +  \phi \right)\phi. \label{Ex4_1c}
\end{eqnarray}
Here, the Riccati system (\ref{rica1})-(\ref{rica6}) admits the solution 
\begin{eqnarray*}
\alpha(t) = \delta(t) = 0, \quad  \beta(t) = \beta(0), \quad  \gamma(t)= t + \gamma(0), \quad   \varepsilon(t) = \varepsilon(0), \ \kappa(t) = \kappa(0), \quad \mu(t) = \mu(0)e^{2t- 2\sech t}.
\end{eqnarray*}
 In these terms, the proposed system (\ref{Ex4_1a})-(\ref{Ex4_1c}) has solutions in the form
\begin{equation}
  \psi(x,t) =  \left(2 + \theta - \frac{a_1}{4}\right) \mu^{-1/2}(0) e^{-t-\sech t+ \kappa(0)} \left[ 1- \tanh\left( \beta(0)x + \varepsilon(0) - \theta (t + \gamma(0)) \right)  \right]^2, 
\end{equation}
\begin{equation}
  \varphi(x,t) = \frac{a_1}{4} \mu^{-1/2}(0) e^{-t-\sech t+ \kappa(0)} \left[ 1 + \tanh\left( \beta(0)x + \varepsilon(0) - \theta (t + \gamma(0)) \right)  \right]^2,
\end{equation}
\begin{equation}
  \phi(x,t) = \left(a_1 - \theta - 2\right) \mu^{-1/2}(0) e^{-t-\sech t+ \kappa(0)}\left[ 1- \tanh\left(  \beta(0)x + \varepsilon(0) - \theta (t + \gamma(0)) \right)  \right],
\end{equation}
provided  $a_1, \ b_2, \ b_3, \ c_1, \ c_3, \ e_1, \ e_2,$ satisfy the same relations established in Section \ref{Sect2}.
\subsection{Solutions for the generalized Gray-Scott model}
This subsection is focused on the study of the generalization of the classical Gray-Scott system. The main result of this part is stated as follows:
\begin{theorem}[\textbf{Generalized Gray-Scott Model}]\label{Th3}
The Gray-Scott model with variable coefficients 
{\footnotesize
\begin{eqnarray}
\psi _{t} &=&a(t)\psi_{xx} -(b\left(
t\right) x^{2}-d(t) - L_{1}(t)-xf(t))\psi -(g(t)-c\left( t\right) x)\psi
_{x} - h_{1}(x,t)\psi \varphi^2 + M_{1}(x,t),  \label{GenGrayScott1} \\
\varphi _{t} &=&a(t)\varphi_{xx} -(b\left(
t\right) x^{2}-d(t) - L_{1}(t)-xf(t))\varphi -(g(t)-c\left( t\right) x)\varphi
_{x}+h_{1}(x,t)\psi \varphi^2+ M_2(x,t), \label{GenGrayScott2} 
\end{eqnarray}%
} can be transformed into the standard  Gray-Scott model
\begin{equation}
  u_{\tau} = u_{\xi \xi} - uv^2 + b_1(1-u),  \label{GrayScott1} 
\end{equation}
\begin{equation}
 v_{\tau} = v_{\xi \xi} + uv^2 - b_1v + b_2. \label{GrayScott2} 
\end{equation}
\end{theorem}

\begin{proof}
As in the previous results, we use the substitutions
\begin{equation}
\psi (x,t)=\frac{1}{\sqrt{\mu (t)}}e^{\alpha (t)x^{2}+\delta (t)x+\kappa
(t)}u(\xi ,\tau ),\qquad \xi =\beta (t)x+\varepsilon (t),\qquad \tau
=\gamma (t),\label{SUBST1}
\end{equation}
and 
\begin{equation}
\varphi (x,t)=\frac{1}{\sqrt{\mu (t)}}e^{\alpha (t)x^{2}+\delta
(t)x+\kappa (t)}v(\xi ,\tau ).\qquad  \label{SUBST2}
\end{equation}
By assuming the coefficients satisfy the Riccati system
(\ref{rica1})-(\ref{rica6}) and imposing the integrability conditions 
$$  h_{1}(x,t)= a(t)\beta^{2}(t)\mu(t)e^{-2(\alpha (t)x^{2}+\delta
(t)x+\kappa (t))}, \quad L_{1}(t) = -b_{1}a(t)\beta^2(t), $$
$$ M_{i}(x,t)= b_{i} a(t)\beta^{2}(t)\mu^{-1/2}(t)e^{\alpha (t)x^{2}+\delta
(t)x+\kappa (t)},  \quad i = 1,2,$$
is easy to see that the functions $u(\xi,\tau), \ v(\xi,\tau),$  satisfy the Gray-Scott model (\ref{GrayScott1})-(\ref{GrayScott2}). 
\end{proof}

\begin{figure}[h!]
\centering
\subfigure[Profile of the solution $\psi$.]{\includegraphics[scale=0.33]{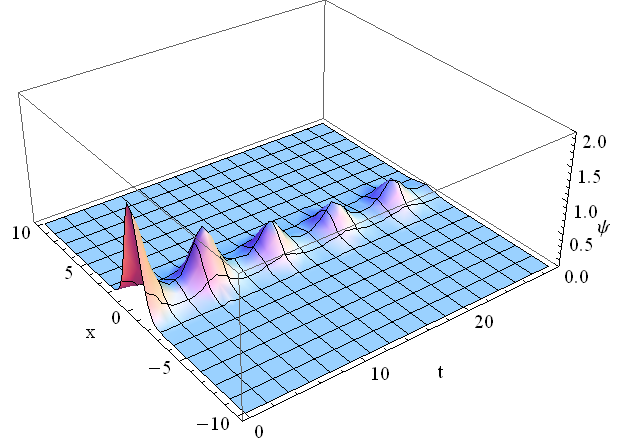}}
\subfigure[Contour of the solution $\psi$.]{\includegraphics[scale=0.36]{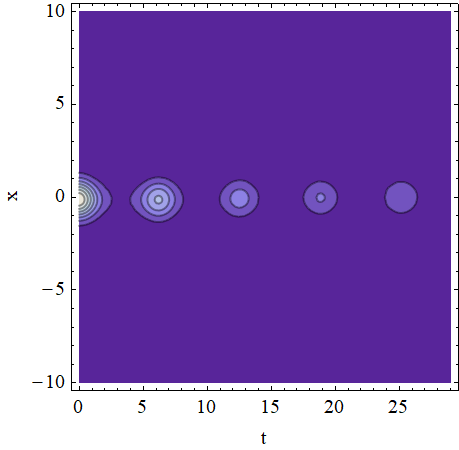}}
\subfigure[Profile of the solution $\varphi$.]{\includegraphics[scale=0.33]{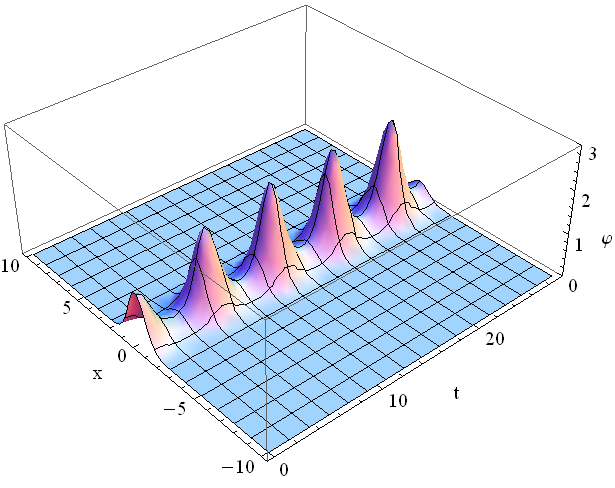}}
\subfigure[Contour of the solution $\varphi$.]{\includegraphics[scale=0.36]{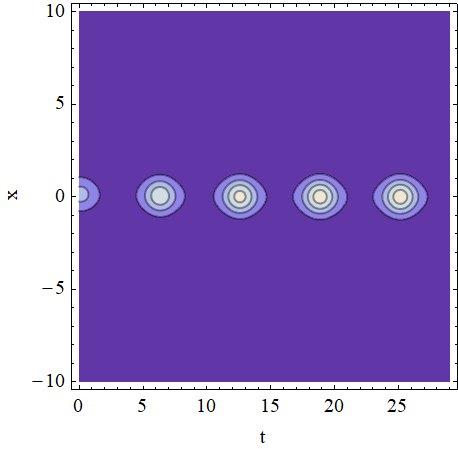}}
\caption{Solutions for the system  (\ref{Ex5a})-(\ref{Ex5b}) for the parameters $ b_1 = \frac{1}{8}, \ b_2 = 0$. Here, (b) and (d) show the contours of $\psi$ and $\varphi$ functions, respectively.} \label{Fig5}
\end{figure}

\subsubsection{\textbf{Constant diffusivity and an interaction coefficient with exponential growth}}
If we assume a system as follows
\begin{eqnarray}
\psi _{t} &=& \psi_{xx} + \left(4x^2 + 2 -\sin t  - b_1 \right)\psi + 4x\psi_x -e^{2x^2 -2\cos t}\psi \varphi^2 + b_1 e^{-x^2 + \cos t},  \label{Ex5a} \\
\varphi _{t} &=& \varphi_{xx} + \left(4x^2 + 2 -\sin t  - b_1 \right)\varphi + 4x\varphi_x +e^{2x^2 -2\cos t}\psi \varphi^2 \label{Ex5b}
\end{eqnarray}
then, Theorem \ref{Th3} allows the construction of solutions in the following closed form:
 {\footnotesize
\begin{equation}
  \psi(x,t) = \exp\left(-x^2 + \cos t\right)\left \{ \frac{3-\sqrt{1-4b_1}}{4} - \frac{\sqrt{2 + 2\sqrt{1-4b_1}-4b_1}}{4}\tanh \left( \frac{\sqrt{1 + \sqrt{1-4b_1}-2b_1}}{4}(x-\theta t) \right) \right \}, 
\end{equation}
\begin{equation}
  \varphi(x,t) = \exp\left(-x^2 + \cos t\right)\left \{ \frac{1+\sqrt{1-4b_1}}{4} + \frac{\sqrt{2 + 2\sqrt{1-4b_1}-4b_1}}{4}\tanh \left( \frac{\sqrt{1 + \sqrt{1-4b_1}-2b_1}}{4}(x-\theta t) \right) \right \},
\end{equation}}
with $\theta = \frac{\sqrt{2}(1-3\sqrt{1-4b_1})}{4}.$ The dynamics of these solutions can be seen in Figure \ref{Fig5}. 

\begin{figure}[h!]
\centering
\subfigure[Profile of the solution $\psi$.]{\includegraphics[scale=0.37]{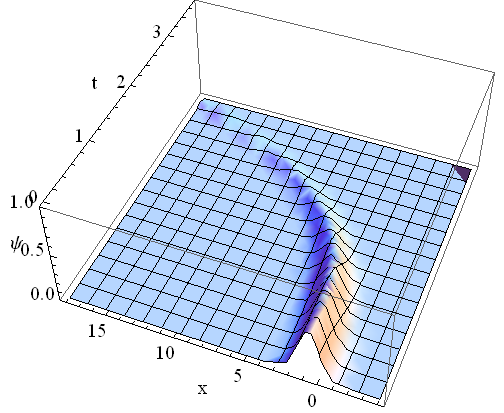}}
\subfigure[Contour of the solution $\psi$.]{\includegraphics[scale=0.37]{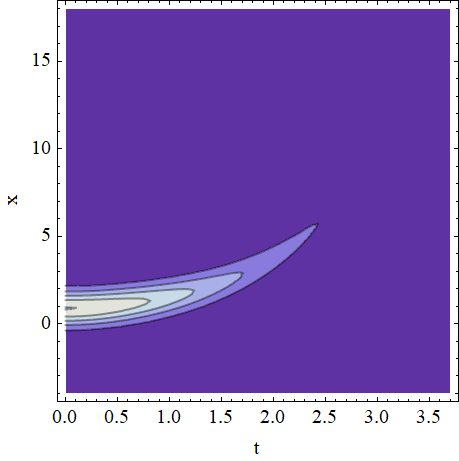}}
\subfigure[Profile of the solution $\varphi$.]{\includegraphics[scale=0.35]{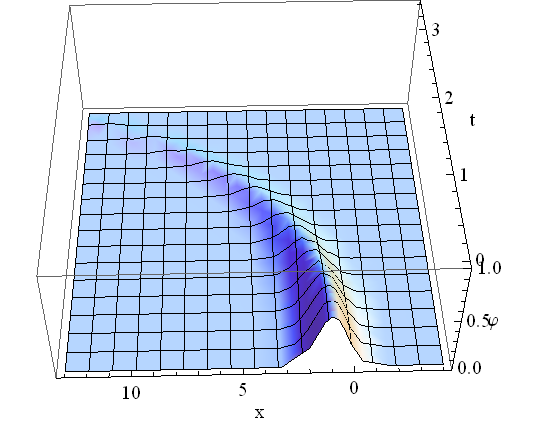}}
\subfigure[Contour of the solution $\varphi$.]{\includegraphics[scale=0.37]{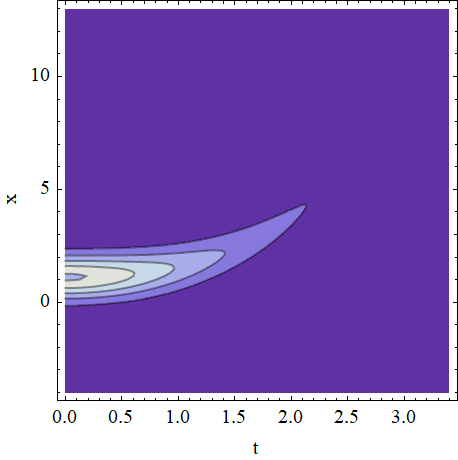}}
\caption{Solutions for the system  (\ref{Ex6a})-(\ref{Ex6b}) for the parameters $ b_1 = \frac{1}{8}, \ b_2 = 0$. The profiles of the functions $\psi$
and $\varphi$ are shown in (a) and (c). Figures (b) and (d) show the contours of $\psi$ and $\varphi$.} \label{Fig6}
\end{figure}

\subsubsection{\textbf{tanh-type diffusivity and an interaction coefficient with exponential growth}}
We consider in this example the general  Gray-Scott model
{\small
\begin{eqnarray}
\psi _{t} = \tanh t\left[\psi_{xx} - 4x^2 \psi + \psi - b_1 \sech^8 t \psi - \sech^6 t e^{2(x-\cosh t)^2}\psi \varphi^2 + b_1 \sech^9 t e^{-(x-\cosh t)^2} \right] \\ \quad \quad \quad \quad  + \sinh t \left(4x\psi -3 \psi_x\right), \label{Ex6a} 
\end{eqnarray}
\begin{eqnarray}
\varphi _{t} = \tanh t\left[\varphi_{xx} - 4x^2 \varphi + \varphi - b_1 \sech^8 t \varphi + \sech^6 t e^{2(x-\cosh t)^2}\psi \varphi^2 + b_1 \sech^9 t e^{-(x-\cosh t)^2} \right] \\ \quad \quad \quad \quad  + \sinh t \left(4x\varphi -3 \varphi_x\right). \label{Ex6b} 
\end{eqnarray}}
By solving the system (\ref{rica1})-(\ref{rica6}), we get
\begin{eqnarray*}
\alpha(t) = -1, \quad  \beta(t) = \sech^4 t, \quad  \gamma(t)= -\frac{\sech^8 t}{8}, \quad   \delta(t) = 2\cosh t, \\ \varepsilon(t) = -\frac{1}{3}\sech^3 t,  \quad \quad  \kappa(t) = -\cosh^2 t, \quad \quad \mu(t) = \cosh^2 t.
\end{eqnarray*}
 Then, the system (\ref{Ex6a})-(\ref{Ex6b}) admits a solution given by the equations
 {\footnotesize
\begin{equation}
  \psi(x,t) = \sech t e^{-(x - \cosh t)^2}\left \{ \frac{3-\sqrt{1-4b_1}}{4} - \frac{\sqrt{2 + 2\sqrt{1-4b_1}-4b_1}}{4}\tanh \left( \frac{\sqrt{1 + \sqrt{1-4b_1}-2b_1}}{4}(\xi-\theta \tau) \right) \right \}, 
\end{equation}
\begin{equation}
  \varphi(x,t) = \sech t e^{-(x - \cosh t)^2}\left \{ \frac{1+\sqrt{1-4b_1}}{4} + \frac{\sqrt{2 + 2\sqrt{1-4b_1}-4b_1}}{4}\tanh \left( \frac{\sqrt{1 + \sqrt{1-4b_1}-2b_1}}{4}(\xi-\theta \tau) \right) \right \},
\end{equation}}
with $\theta = \frac{\sqrt{2}(1-3\sqrt{1-4b_1})}{4},$ $\xi = x\sech^4 t - \frac{1}{3}\sech^3 t$, and $\tau = -\frac{1}{8}\sech^8 t.$  In this situation, we produce solutions with central axes bent to the right and decreasing in amplitude over time, see Figure \ref{Fig6}. This fascinating phenomenon results of the modulation of the hyperbolic cosine and the exponential decay induced by the interaction coefficient.

\subsection{Solutions for the generalized Burgers system }
This section concludes with the introduction of a novel coupled Burgers system that generalizes on previous models reported in the literature. In this context, the simplest version of the Riccati system must be used in the construction of the explicit solutions for such a model. To be more exact, the result says the following:
\begin{theorem}[\textbf{Generalized Burgers System}]\label{Th4}
The  variable coefficient coupled Burgers system 
\begin{eqnarray}
\psi _{t} &=&a(t)\psi_{xx} - b_{1}a(t)\psi \psi_x -c_{1}a(t)(\psi \varphi)_{x} + c(t)\left(\psi + x\psi_{x}\right) -g(t) \psi_{x},  \label{GBE1} \\
\varphi _{t} &=&a(t)\varphi_{xx}-b_{2}a(t)\varphi \varphi_{x} -c_{2}a(t)(\psi \varphi)_{x} + c(t)\left(\varphi + x\varphi_{x}\right) -g(t)\varphi_{x}  \label{GBE2}
\end{eqnarray}%
can be transformed into the Burgers system
\begin{equation}
  u_{\tau} = u_{\xi \xi} - b_{1}uu_{\xi}- c_{1}(uv)_{\xi}  \label{BE1}
\end{equation}
\begin{equation}
 v_{\tau} = v_{\xi \xi} -b_{2}vv_{\xi}- c_{2}(uv)_{\xi}.  \label{BE2}
\end{equation}
\end{theorem}

\begin{proof}
Consider the substitutions 
\begin{equation}
\psi (x,t) = \beta(t)u(\xi ,\tau ),\qquad \xi =\beta (t)x+\varepsilon (t),\qquad \tau
=\gamma (t),\label{SubsBE1}
\end{equation}
and
\begin{equation}
  \varphi (x,t) = \beta(t)v(\xi ,\tau ). \label{SubsBE2}
\end{equation}
Now, computing the first derivatives of $\psi$ (similarly the derivatives of $\varphi$ are obtained by replacing $u$ by $v$) we have:
\begin{equation}
 \psi_x = \beta^{2}(t)u_{\xi}, \quad \quad  \psi_{xx} = \beta^{3}(t)u_{\xi \xi},  \label{Deri_1}
\end{equation}
\begin{equation}
  \psi_t = \beta^{\prime}(t) u + \beta(t) \left[u_{\xi}\left(\beta^{\prime}(t)x + \varepsilon^{\prime}(t) \right)+ u_{\tau}\gamma^{\prime}(t)\right]. \label{Deri_2}  
\end{equation}
Inserting the last equations (\ref{Deri_1})-(\ref{Deri_2}) into the system (\ref{GBE1})-(\ref{GBE2}), one obtains the Burgers system (\ref{BE1})-(\ref{BE2}) by imposing the conditions 
\begin{equation}
 \frac{d \beta}{dt} = c(t)\beta(t), \quad \quad   \frac{d \gamma}{dt} = a(t)\beta^{2}(t), \quad \quad \frac{d \varepsilon}{dt} = -g(t)\beta(t). \label{ReducedRiccati}
\end{equation}
We point out that the system of ODEs (\ref{ReducedRiccati})  is a collapsed version with $\alpha = \delta = \kappa = b = f = 0$ of the classical Riccati system  (\ref{rica1})-(\ref{rica6}).
\end{proof}

\begin{figure}[h!]
\centering
\subfigure[Profile of the solution $\psi$.]{\includegraphics[scale=0.33]{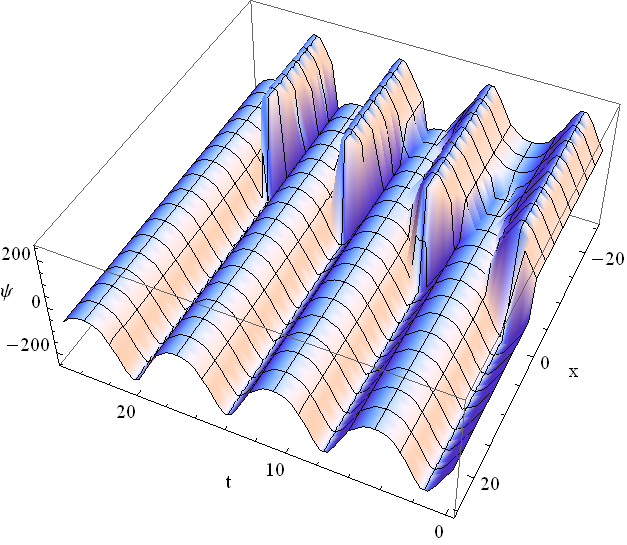}}
\subfigure[Contour of the solution $\psi$.]{\includegraphics[scale=0.33]{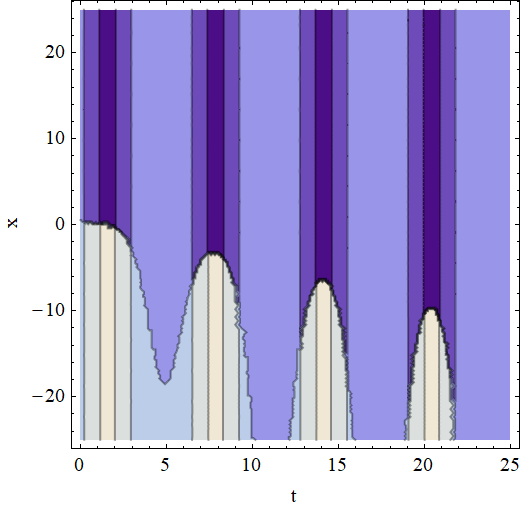}}
\subfigure[Profile of the solution $\varphi$.]{\includegraphics[scale=0.33]{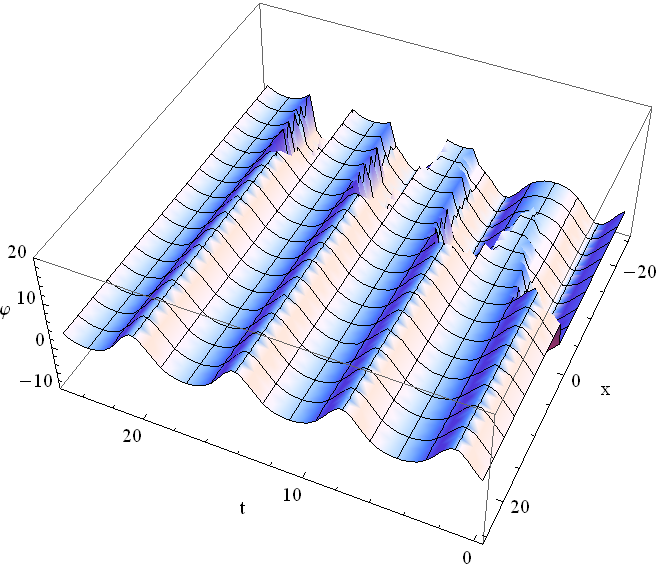}}
\subfigure[Contour of the solution $\varphi$.]{\includegraphics[scale=0.35]{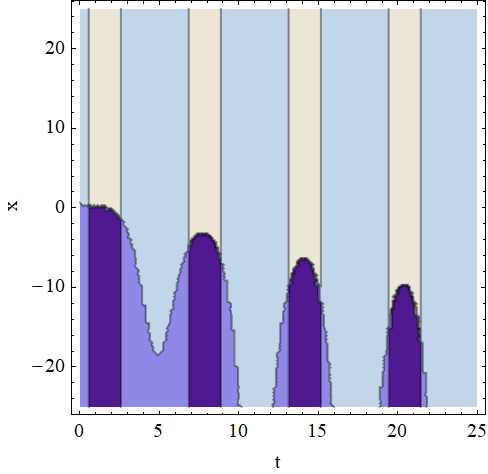}}
\caption{Solutions for the system  (\ref{Ex7a})-(\ref{Ex7b}) for the parameters $ b_1 = 2, \ b_2 = 1, \ c_1 = c_2 = 2$, $A = -\frac{111}{8},$ $B = -1,$ $\beta(0) = 1,$ $\gamma(0) = \varepsilon(0) = 0$.} \label{Fig7}
\end{figure}

In the preceding theorem, the integrability of the Burgers system is obtained by equalizing the diffusion term and the interaction coefficient, $a(t)$. Given the nature of the solutions expressed in (\ref{SubsBE1})-(\ref{SubsBE2}), $c(t)$ has an important influence on its dynamics (via the function $\beta(t)$), as it affects horizontal translations and amplitudes. The following examples provide insight into the dynamics of solutions when $c(t)$ has periodic structures.
\subsubsection{\textbf{Periodic diffusivity and interaction terms}}
In the present example we will consider the Burgers system
\begin{eqnarray}
\psi _{t} &=& \beta(0)^{-2}e^{-2\sin t}\left(\psi_{xx} - b_1 \psi \psi_x -c_1(\psi \varphi)_x \right) + \cos t \left( \psi + x \psi_x - \psi_x \right),  \label{Ex7a} \\
\varphi _{t} &=& \beta(0)^{-2}e^{-2\sin t}\left(\varphi_{xx} - b_2 \varphi \varphi_x -c_2(\psi \varphi)_x \right) + \cos t \left( \varphi + x \varphi_x - \varphi_x \right). \label{Ex7b}
\end{eqnarray}

\begin{figure}[h!]
\centering
\subfigure[Profile of the solution $\psi$.]{\includegraphics[scale=0.33]{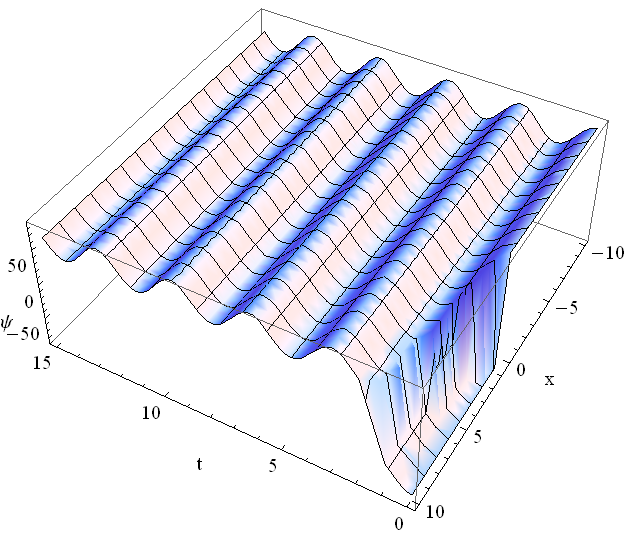}}
\subfigure[Contour of the solution $\psi$.]{\includegraphics[scale=0.34]{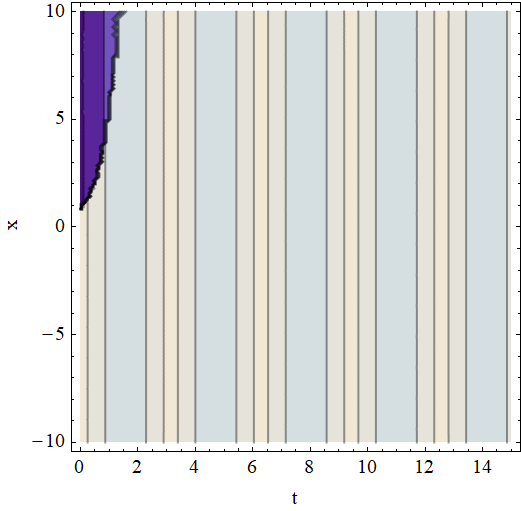}}
\subfigure[Profile of the solution $\varphi$.]{\includegraphics[scale=0.33]{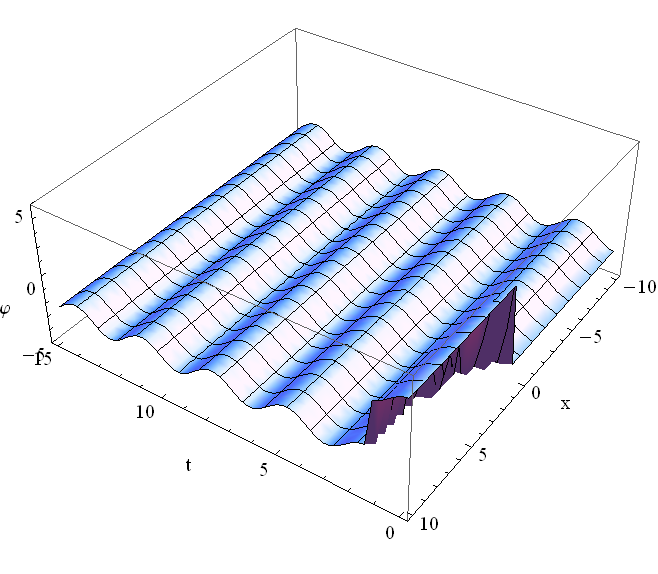}}
\subfigure[Contour of the solution $\varphi$.]{\includegraphics[scale=0.36]{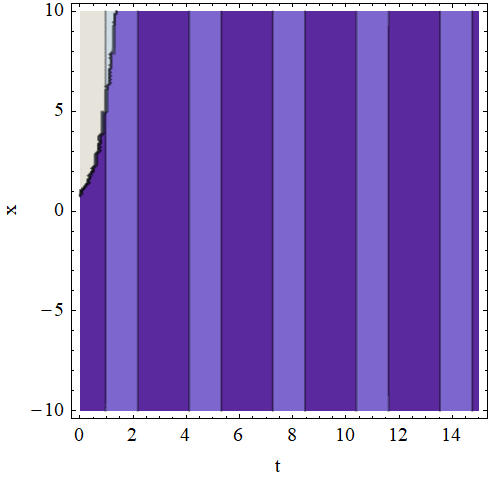}}
\caption{Solutions for the system  (\ref{Ex8a})-(\ref{Ex8b}) for the parameters $ b_1 = 2, \ b_2 = 1, \ c_1 = c_2 = 2$,  $A = \frac{111}{8},$ $B = 1,$ $\beta(0) = \gamma(0) = \varepsilon(0) = 1$.} \label{Fig8}
\end{figure}
The solution of the system (\ref{ReducedRiccati}) is given by the functions
\begin{eqnarray*}
 \beta(t) =  \beta(0)e^{\sin t}, \quad  \quad \gamma(t)= t + \gamma(0), \quad  \quad  \varepsilon(t) = \varepsilon(0) -\beta(0)\left(e^{\sin t}-1\right).
\end{eqnarray*}
 Then, we can construct  solutions for the system (\ref{Ex7a})-(\ref{Ex7b}) as follows
 {\footnotesize
\begin{equation}
  \psi = \beta(0)e^{\sin t}\left \{ B -2A\left(2c_1 -\frac{1}{4c_{1}c_{2}} -1\right)\tanh \left[A\left(20(\beta(0)e^{\sin t}(x-1)+ \beta(0) + \varepsilon(0)) -10 -2A(t + \gamma(0)) \right)\right]\right \}, 
\end{equation}
\begin{equation}
  \varphi = \beta(0)Be^{\sin t}\left(2c_1 - \frac{1}{4c_{1}c_{2}} -1\right)\tanh \left[A\left(20(\beta(0)e^{\sin t}(x-1)+ \beta(0) + \varepsilon(0)) -10 -2A(t + \gamma(0)) \right)\right]. 
\end{equation}
}

The dynamics of these solutions are shown in Figure \ref{Fig7}.

\subsubsection{\textbf{Diffusivity and interaction terms with exponential growth}}
Consider the system of equations
{\footnotesize
\begin{eqnarray}
\psi _{t} &=& \frac{\gamma(0)e^t}{\beta^2(0)\cos^2(\sin t)}\left(\psi_{xx} - b_1 \psi \psi_x -c_1(\psi \varphi)_x \right) - \tan(\sin t)\cos t \left( \psi + x \psi_x\right)-\frac{\varepsilon(0)\sin t}{\beta(0)\cos(\sin t)}\psi_x ,  \label{Ex8a} \\
\varphi _{t} &=& \frac{\gamma(0)e^t}{\beta^2(0)\cos^2(\sin t)}\left(\varphi_{xx} - b_2 \varphi \varphi_x -c_2(\psi \varphi)_x \right) - \tan(\sin t)\cos t \left( \varphi + x \varphi_x\right)-\frac{\varepsilon(0)\sin t}{\beta(0)\cos(\sin t)}\varphi_x. \label{Ex8b}
\end{eqnarray}}
Such a Burgers system admits the functions
\begin{eqnarray*}
 \beta(t) =  \beta(0)\cos(\sin t), \quad  \quad \gamma(t)=  \gamma(0)e^{t}, \quad  \quad  \varepsilon(t) = \varepsilon(0)\cos t.
\end{eqnarray*}
 Then, solutions can be constructed in the following  form:
 {\footnotesize
\begin{equation}
  \psi = \beta(0)\cos(\sin t)\left \{ B -2A\left(2c_1 -\frac{1}{4c_{1}c_{2}} -1\right)\tanh \left[A\left(20(\beta(0)\cos(\sin t)x+\varepsilon(0)\cos t) -10 -2A\gamma(0)e^{t}\right)\right]\right \}, 
\end{equation}
\begin{equation}
  \varphi = \beta(0)\cos(\sin t)\left(2c_1 - \frac{1}{4c_{1}c_{2}} -1\right)\tanh \left[A\left(20(\beta(0)\cos(\sin t)x+\varepsilon(0)\cos t) -10 -2A\gamma(0)e^{t}\right)\right]. 
\end{equation}}

Figure \ref{Fig8} shows the profiles of these solutions. 

\section{Conclusions and Final Remarks}
\label{Sect4}

In this paper, we introduced and investigated the explicit solutions of a coupled reaction-diffusion system and a coupled  Burgers-type system with variable coefficients. The structure of such  systems extends the classical linear reaction-diffusion model, the diffusive Lotka-Volterra  system, the Gray-Scott model, and the Burgers equations. We showed that if the coefficients fulfill a Riccati system, the general models possess solitary wave solutions as well as solutions with nontrivial behavior (in relation to solitary waves), presenting bending properties. 

We point out that exact solutions always play a crucial role for any nonlinear PDE (or system of PDEs) describing real-world processes. In fact, even exact solutions with questionable applications can be used as test problems to estimate the accuracy and efficiency of numerical methods. As a result, our research makes a significant contribution in this aspect because the level of trust in a numerical method for approximating thus general system of equations increases if it can accurately replicate the dynamics of more complex solutions.  

On the other hand, it is important to mention that the results reported here apply to any type of solution of the standard reaction-diffusion and Burgers systems. Likewise, by appropriately selecting the functions $F_i$ in the system (\ref{Gen1})-(\ref{Gen2}), it can include new reaction-diffusion models such as the Brusselator model \cite{RajniRohila2016,Nicolis1989}, the  isothermal chemical system \cite{Garcia1996,RajniRohila2016},  and the Noyes-Field model \cite{RODRIGO2001}. Therefore, the ideas proposed in this paper can be applicable to these new models.

\begin{acknowledgement}
J. M. Escorcia thanks Universidad EAFIT  for the financial support provided for this research (Internal Project No. 12330022023). 
\end{acknowledgement}

\section{Appendix: Solutions of the Riccati Systems }
\label{Sect5}
In this appendix, we present the solutions of the Riccati systems used in the construction of the explicit solutions of the various reaction-diffusion equations. 
\subsection{Solution of the Riccati system (\ref{Rica_1})-(\ref{Carac__1})}
A solution of this Riccati system including multiparameters is given by the following expressions \cite{CorderoSoto2008,Escorcia,Suazo009,Suslov12}:

\begin{equation}
\mu \left( t\right) =-2\mu \left( 0\right) \mu _{0}\left( t\right) \left(
\alpha \left( 0\right) +\gamma _{0}\left( t\right) \right) ,  \label{mu}
\end{equation}%
\begin{equation}
\alpha \left( t\right) =\alpha _{0}\left( t\right) -\frac{\beta
_{0}^{2}\left( t\right) }{4\left( \alpha \left( 0\right) +\gamma _{0}\left(
t\right) \right) },  \label{alpha}
\end{equation}%
\begin{equation}
\beta \left( t\right) =-\frac{\beta \left( 0\right) \beta _{0}\left(
t\right) }{2\left( \alpha \left( 0\right) +\gamma _{0}\left( t\right)
\right) },  \label{beta}
\end{equation}%
\begin{equation}
\gamma \left( t\right) =\gamma \left( 0\right) -\frac{\beta
^{2}\left( 0\right) }{4\left( \alpha \left( 0\right) +\gamma _{0}\left(
t\right) \right) },  \label{gamma}
\end{equation}%
\begin{equation}
\delta \left( t\right) =\delta _{0}\left( t\right) -\frac{\beta _{0}\left(
t\right) \left( \delta \left( 0\right) +\varepsilon _{0}\left( t\right)
\right) }{2\left( \alpha \left( 0\right) +\gamma _{0}\left( t\right) \right) 
},  \label{delta}
\end{equation}%
\begin{equation}
\varepsilon \left( t\right) =\varepsilon \left( 0\right) -\frac{\beta \left(
0\right) \left( \delta \left( 0\right) +\varepsilon _{0}\left( t\right)
\right) }{2\left( \alpha \left( 0\right) +\gamma _{0}\left( t\right) \right) 
},  \label{epsilon}
\end{equation}%
\begin{equation}
\kappa \left( t\right) =\kappa \left( 0\right) +\kappa _{0}\left( t\right) -%
\frac{\left( \delta \left( 0\right) +\varepsilon _{0}\left( t\right) \right)
^{2}}{4\left( \alpha \left( 0\right) +\gamma _{0}\left( t\right) \right) },
\label{kappa}
\end{equation}%
\ subject to the initial arbitrary conditions $\mu \left( 0\right) ,$ $%
\alpha \left( 0\right) ,$ $\beta \left( 0\right) \neq 0,$ $\gamma (0),$ $%
\delta (0),$ $\varepsilon (0)$ and $\kappa (0)$. Here, $\alpha _{0}$, $\beta _{0}$%
, $\gamma _{0}$, $\delta _{0}$, $\varepsilon _{0}$ and $\kappa _{0}$ are
given explicitly by\ 
\begin{equation}
\alpha _{0}\left( t\right) =-\frac{1}{4a\left( t\right) }\frac{\mu
_{0}^{\prime }\left( t\right) }{\mu _{0}\left( t\right) }-\frac{d\left(
t\right) }{2a\left( t\right) },  \label{alpha0}
\end{equation}%
\begin{equation}
\beta _{0}\left( t\right) =\frac{W\left( t\right) }{\mu _{0}\left( t\right) 
},\quad W\left( t\right) =\exp \left( \int_{0}^{t}\left( c\left( s\right)
-2d\left( s\right) \right) \ ds\right) ,  \label{beta0}
\end{equation}%
\begin{equation}
\gamma _{0}\left( t\right) =-\frac{1}{2\mu _{1}\left( 0\right) }\frac{\mu
_{1}\left( t\right) }{\mu _{0}\left( t\right) }+\frac{d\left( 0\right) }{%
2a\left( 0\right) },  \label{gamma0}
\end{equation}%
\begin{equation}
\delta _{0}\left( t\right) =\frac{W\left( t\right) }{\mu _{0}\left( t\right) 
}\ \ \int_{0}^{t}\left[ \left( f\left( s\right) +\frac{d\left( s\right) }{%
a\left( s\right) }g\left( s\right) \right) \mu _{0}\left( s\right) +\frac{%
g\left( s\right) }{2a\left( s\right) }\mu _{0}^{\prime }\left( s\right) %
\right] \ \frac{ds}{W\left( s\right) },  \label{delta0}
\end{equation}%
\begin{eqnarray}
\varepsilon _{0}\left( t\right) &=&-\frac{2a\left( t\right) W\left( t\right) 
}{\mu _{0}^{\prime }\left( t\right) }\delta _{0}\left( t\right)
-8\int_{0}^{t}\frac{a\left( s\right) \sigma \left( s\right) W\left( s\right) 
}{\left( \mu _{0}^{\prime }\left( s\right) \right) ^{2}}\left( \mu
_{0}\left( s\right) \delta _{0}\left( s\right) \right) \ ds  \label{epsilon0}
\\
&&+2\int_{0}^{t}\frac{a\left( s\right) W\left( s\right) }{\mu _{0}^{\prime
}\left( s\right) }\left[ f\left( s\right) +\frac{d\left( s\right) }{a\left(
s\right) }g\left( s\right) \right] \ ds,  \notag
\end{eqnarray}%
\begin{eqnarray}
\kappa _{0}\left( t\right) &=&-\frac{a\left( t\right) \mu _{0}\left( t\right) 
}{\mu _{0}^{\prime }\left( t\right) }\delta _{0}^{2}\left( t\right)
-4\int_{0}^{t}\frac{a\left( s\right) \sigma \left( s\right) }{\left( \mu
_{0}^{\prime }\left( s\right) \right) ^{2}}\left( \mu _{0}\left( s\right)
\delta _{0}\left( s\right) \right) ^{2}\ ds  \label{kappa0} \\
&&\quad +2\int_{0}^{t}\frac{a\left( s\right) }{\mu _{0}^{\prime }\left(
s\right) }\left( \mu _{0}\left( s\right) \delta _{0}\left( s\right) \right) %
\left[ f\left( s\right) +\frac{d\left( s\right) }{a\left( s\right) }g\left(
s\right) \right] \ ds,  \notag
\end{eqnarray}%
\ with $\delta _{0}\left( 0\right) =g_{0}\left( 0\right) /\left( 2a\left(
0\right) \right) ,$ $\varepsilon _{0}\left( 0\right) =-\delta _{0}\left(
0\right) ,$ $\kappa _{0}\left( 0\right) =0.$ Here $\mu _{0}$ and $\mu _{1}$
represent the fundamental solution of the characteristic equation subject to
the initial conditions $\mu _{0}(0)=0$, $\mu _{0}^{\prime }(0)=2a(0)\neq 0$
and $\mu _{1}(0)\neq 0$, $\mu _{1}^{\prime }(0)=0$.

We point out that all the formulas involved in such a solution have been verified previously in \cite{Ko-su-su}. 

\subsection{Solution of the modified Riccati system (\ref{Ricati1})-(\ref{Ricati7})}
In order to find the solution of this Riccati system, we just need to solve the equations for the functions $\kappa_1(t)$ and $\kappa_2(t) $ because the rest of the equations have been solved previously. Suppose $\kappa(t) = \kappa_1(t) + \kappa_2(t).$ Then, the function $\kappa(t)$ satisfies the equation (\ref{Rica_6}) of the earlier Riccati system, i.e.,  
$$ \frac{d \kappa}{dt} = a(t)\delta^2(t)-g(t)\delta(t).   $$
On other hand, the equation for $\kappa_2(t)$ can be solved explicitly:
$$ \kappa_2(t) = -\ln \Big|\int_{0}^{t}h(s) \ ds + e^{-\kappa_2(0)} \Big|. $$
In these terms, using equation (\ref{kappa}), we have
$$ \kappa_1(t) = \kappa(t)-\kappa_2(t) = \kappa_1 \left( 0\right) + \kappa_2 \left( 0\right) +\kappa _{0}\left( t\right) -
\frac{\left( \delta \left( 0\right) +\varepsilon _{0}\left( t\right) \right)
^{2}}{4\left( \alpha \left( 0\right) +\gamma _{0}\left( t\right) \right) } + \ln \Big|\int_{0}^{t}h(s)ds + e^{-\kappa_2(0)} \Big|    ,$$
and we solved completely the system (\ref{Ricati1})-(\ref{Ricati7}).

\bibliography{references}
\bibliographystyle{plain} 


\end{document}